\documentclass{article}
\usepackage[utf8]{inputenc}
\usepackage[margin = 1in]{geometry}

\usepackage{booktabs} 
\usepackage[ruled]{algorithm2e} 

\SetAlFnt{\small}
\SetAlCapFnt{\small}
\SetAlCapNameFnt{\small}
\SetAlCapHSkip{0pt}
\IncMargin{-\parindent}

\usepackage{tabularx}
\usepackage{xspace}
\usepackage{mathtools}
\usepackage{url}
\usepackage{subcaption}
\usepackage{tikz}
\usetikzlibrary{arrows, fit}
\usetikzlibrary{backgrounds,automata,calc}
\usetikzlibrary{decorations.pathreplacing}
\usepackage{csquotes}
\usepackage{amsmath,amsfonts,amsthm}
\usepackage{pgfplots}
\pgfplotsset{compat=1.17} 
\DeclareMathOperator*{\argmax}{arg\,max}

\usepackage{enumerate}
\usepackage{paralist}
\usepackage[shortlabels]{enumitem}
\usepackage{appendix}
\usepackage[capitalize]{cleveref}
\usepackage{thmtools}
\usepackage{mathtools}
\usepackage{natbib}
\usepackage{changepage}

\newcommand{\MAXUSW}{\texttt{MAX-USW}\xspace}

\newcommand{\MMS}{\texttt{MMS}\xspace}
\newcommand{\MWNW}{\texttt{MWNW}\xspace}

\newcommand{\R}{\mathbb{R}}
\newcommand{\Z}{\mathbb{Z}}

\newcommand{\FindDesired}{\texttt{Find-Desired}}
\newcommand{\GetDistances}{\texttt{Get-Distances}}
\newcommand{\prev}{\texttt{prev}}

\renewcommand{\cal}[1]{\mathcal{#1}}

\newtheorem{theorem}{Theorem}[section]
\newtheorem{lemma}[theorem]{Lemma}

\newtheorem{prop}[theorem]{Proposition}
\newtheorem{corr}[theorem]{Corollary}

\theoremstyle{definition}

\theoremstyle{definition}
\newtheorem{definition}[theorem]{Definition}

\newtheorem{example}[theorem]{Example}

\usepackage{appendix}
\usepackage[capitalize]{cleveref}

\title{A General Framework for Fair Allocation under Matroid Rank Valuations}
\author{Vignesh Viswanathan and Yair Zick \\
University of Massachusetts, Amherst \\
{\texttt{\{vviswanathan, yzick\}@umass.edu}}}
\date{}

\begin{document}

\maketitle

\begin{abstract}
We study the problem of fairly allocating a set of indivisible goods among agents with matroid rank valuations --- every good provides a marginal value of $0$ or $1$ when added to a bundle and valuations are submodular.
We generalize the Yankee Swap algorithm to create a simple framework, called {\em General Yankee Swap}, that can efficiently compute allocations that maximize any justice criterion (or fairness objective) satisfying some mild assumptions. 
Along with maximizing a justice criterion, General Yankee Swap is guaranteed to maximize utilitarian social welfare, ensure strategyproofness and use at most a quadratic number of valuation queries. 
We show how General Yankee Swap can be used to compute allocations for five different well-studied justice criteria:
\begin{inparaenum}[(a)]
    \item Prioritized Lorenz dominance,
    \item Maximin fairness,
    \item Weighted leximin, 
    \item Max weighted Nash welfare, and
    \item Max weighted $p$-mean welfare.
\end{inparaenum}
In particular, our framework provides the first polynomial time algorithms to compute weighted leximin, max weighted Nash welfare and max weighted $p$-mean welfare allocations for agents with matroid rank valuations.
\end{abstract}

\section{Introduction}\label{sec:intro}
Consider the problem of assigning people to activities (say, a community center offering after-school activities, or university students signing up for classes). People have preferences over the activities they want to take; however, there is a limited number of space available in each activity. 
The assignment must satisfy several additional constraints; for example, people may not sign up for activities with conflicting schedules, and there may also be limits on the maximal number of activities one can sign up for. 
In addition, people often have \emph{priorities}: community center members should be offered priority over non-members; similarly, students closer to graduation may be given priority over first-years when choosing electives.
This problem can be naturally modeled as an instance of a \emph{fair allocation} problem: people are agents who have a \emph{utility} for receiving bundles of \emph{items} (activities). 
Our objective is to efficiently compute an \emph{allocation} (an assignment of activities to people) that satisfies certain \emph{justice criteria}. 
For example, one might be interested in computing an \emph{efficient} allocation --- in the activity allocation setting, this would be an allocation that maximizes the number of activities assigned to people who are willing and able to take them. 
Alternatively, one might want to find an \emph{envy-free} assignment --- one where every person prefers their assigned set of activities to that of any other person.  
Finding these fair and efficient allocations is computationally intractable under general agent utilities \citep{bouveret2016handbook}. 
However, we assume that agent preferences follow a combinatorial structure: for example, under mild assumptions, agent preferences over activities are \emph{submodular}; in economic jargon, they have decreasing returns to scale. 
In addition, assuming people are only interested in taking the maximal number of activities relevant to them, their gain from taking an additional activity is either $0$ or $1$. 
This class of preferences is known as \emph{binary submodular} valuations \citep{benabbou2021MRF}. 

Recently, \citet{viswanathan2022yankee} introduced the \emph{Yankee Swap} algorithm for computing fair and efficient allocations. 
The algorithm starts with all items unassigned, and proceeds in rounds; at every round, an agent is picked to play. The agent can choose to either pick an unassigned item to add to its bundle, or to steal an item from another agent. If they choose to steal, then they initiate a \emph{transfer path}, where each agent steals an item from another, until the final agent takes an unassigned item. This continues until no agents wants to take any unassigned item.

If agents have \emph{binary submodular} valuations, Yankee Swap is guaranteed to output a \emph{Lorenz dominating} allocation. 
When agents have binary submodular valuations, Lorenz dominating allocations are extremely appealing \cite{Babaioff2021Dichotomous}: they 
\begin{inparaenum}
\item maximize both utilitarian and Nash welfare; 
\item are envy-free up to any item (EFX); 
\item offer each agent at least half their maximin share and 
\item can be computed truthfully.
\end{inparaenum}
We explore an extension of this framework where agents have \emph{weights} (also known as \emph{entitlements} or \emph{priorities}), as proposed by \citet{chakraborty2021weighted}; that is, some agents are intrinsically more important than others, and should be given higher priority as a function of their weight.
When agents have weights, Lorenz dominating allocations lose their appeal \cite{chakraborty2021pickingsequences}: they no longer offer approximate envy-freeness, MMS or Nash welfare guarantees. 
In fact, \citet{chakraborty2021pickingsequences} argue against their use in weighted settings.
While most solution concepts in the unweighted setting have been shown to be efficiently computable \citep{Barman2021MRFMaxmin,Babaioff2021Dichotomous}, little is known about the weighted setting.
Our main result bridges this gap: despite the incompatibility of the different solutions in the weighted domain, we propose a 
\begin{displayquote}
\textit{fast, simple and easily explainable algorithm that can compute a variety of solutions in weighted settings and beyond.}
\end{displayquote}


\subsection{Our Contribution}
We present a surprising algorithmic result: a minor modification to the Yankee Swap algorithm --- the order in which it lets agents play --- allows it to compute allocations satisfying a broad range of justice criteria. 
Justice criteria (denoted by $\Psi$) can be thought of as ways to compare allocations; for example, an allocation $X$ is better than $Y$ according the Nash social welfare criterion if the product of agent utilities under $X$ is greater than the product of utilities under $Y$. 

The General Yankee Swap framework picks an agent at each iteration according to a general {\em gain function} $\phi$. 
The gain function $\phi$ takes as input an agent $i$ and an allocation $X$ and outputs a `score' corresponding to the `gain' of adding a good to $i$ according to $\Psi$. 
The selected agent can either take an unassigned item they like, or steal an item from someone else. 
If they steal, then the agent who had an item stolen from them gets to either take an unassigned item or steal yet another item. This goes on until some agent picks an unassigned item they like. If no such path exists, the selected agent is removed from play.  

General Yankee Swap computes a $\Psi$ maximizing allocation for any $\Psi$ as long as the following two conditions hold:
\begin{inparaenum}[(a)]
    \item $\Psi$ maximizing allocations are Pareto dominant, and 
    \item $\Psi$ admits a coherent gain function $\phi$ such that $\phi(X, i)$ is a decreasing function of the utility of agent $i$ under the allocation $X$.
\end{inparaenum}
Several well-known justice criteria satisfy these conditions, summarized in Table \ref{table:results-summary}. 

\begin{table}
\centering
\begin{tabularx}{\textwidth} {>{\hsize=.6\hsize}XX}
\toprule
\textbf{Justice Criterion}&\textbf{Gain Function $\phi(X,i)$}\\
\toprule
 Lorenz Dominance (Thm. \ref{thm:lorenz-dominance}) & $-v_i(X_i)$ \\
 \midrule
 Weighted Leximin (Thm. \ref{thm:weighted-leximin}) & $-\frac{v_i(X_i)}{w_i}$ \\
 \midrule
 Maximin Fair Share (Corr. \ref{corr:MMS-MRF})& $-\frac{v_i(X_i)}{\MMS_i}$ if $\MMS_i>0$, and $-\infty$ otherwise. \\
 \midrule
 Weighted Nash (Thm. \ref{thm:weighted-nash})& $\big (1 + \frac1{v_i(X_i)}\big )^{w_i}$ if $v_i(X_i)>0$; a large $M$ otherwise \\
 \midrule
Weighted $p$-Mean (Thm. \ref{thm:weighted-p-mean})& $\texttt{sign}(p) \times w_i \times [(v_i(X_i) + 1)^p - v_i(X_i)^p]$\\
 \midrule
Weighted Harmonic \cite{montanari2022weightedenvy} & $\frac{w_i}{v_i(X_i) + 1}$ \\
 \bottomrule
\end{tabularx}
\caption{A summary of justice criteria for which General Yankee Swap (\Cref{algo:weighted-yankee-swap}) computes an optimal solution. $v_i(X_i)$ denotes the utility of agent $i$ under the allocation $X$, $w_i$ denotes the weight of agent $i$ and $\MMS_i$ denotes the maximin share of agent $i$.}
\label{table:results-summary}
\end{table}
In addition to weighted settings, General Yankee Swap can produce an allocation that guarantees every agent a maximal fraction of their maximin share, \emph{assuming that we have precomputed the maximin share of every agent $i$, $\MMS_i$}.  
Furthermore, the allocations that General Yankee Swap outputs always maximize utilitarian social welfare, and are always truthful. 
Thus, despite its simplicity, General Yankee Swap offers a unified framework for computing various optimal allocations, under different definitions of optimality.

\subsection{Our Techniques}
There are two broad techniques used in this paper. 
The first is that of {\em path augmentations} in the item exchange graph.
The item exchange graph is a directed graph over the set of goods where an edge exists from one good $g$ to another good $g'$ if the agent who is currently allocated $g$ would be indifferent to swapping the good $g$ with $g'$. Transferring goods along a path in the exchange graph can be used to manipulate allocations and improve them; this is a special case of a general optimization technique known as path augmentation \cite{schrijver-book}.
Path augmentations have been extensively used in prior work \citep{viswanathan2022yankee,Barman2021MRFMaxmin}. 
However, prior works restrict their attention to one specific justice criterion. 
In our work, we exploit the power of path augmentations even further, and apply this technique to a broad range of justice criteria. 
We also use path augmentations to show that our framework is truthful --- adopting a proof style from \citet{Babaioff2021Dichotomous}.

Our second technique is a careful combination of binary search and breadth first search for the efficient computation of paths on the exchange graph. 
This technique allows us to find paths on the exchange graph without explicitly constructing it. 
It was first introduced by \citet{Chakrabarty2019MatroidIntersection} to create fast algorithms for the matroid intersection problem.
We use it to improve the runtime of our algorithm from cubic valuation queries \citep{viswanathan2022yankee} to quadratic valuation queries (ignoring logarithmic factors). 

\subsection{Related Work}
Several works explore fair allocation under binary submodular valuations. \citet{benabbou2019group} present algorithms that compute fair and efficient allocations when agents have binary matching-based (OXS) valuations, a result later extended to general binary submodular valuations by \citet{benabbou2021MRF}.  \citet{Babaioff2021Dichotomous} extend this work further, showing that it is possible to \emph{truthfully} compute Lorenz dominating allocations in polynomial time; Lorenz dominance is a stronger notion of fairness than leximin and it implies a host of other fairness properties. \citet{barman2022groupstrategyproof} add to this result, showing that Lorenz dominating allocations can be computed in a {\em group} strategyproof manner.
\citet{Barman2021MRFMaxmin} show that it is possible to compute an allocation which guarantees each agent their maximin share. 
Our results use some technical lemmas from their work.

Our work also contributes to the existing literature on fair allocation with asymmetric agents. 
Several weighted fairness metrics are proposed in the literature. \citet{farhadi2019wmms} introduce a weighted notion of the maximin share. 
\citet{chakraborty2021weighted} introduce the notion of weighted envy-freeness and propose algorithms to compute weighted envy-free up to one item. 
These results are extended by \citet{chakraborty2022generalized} and \citet{montanari2022weightedenvy}. 
\citet{azis2020wprop} introduced the notion of weighted proportionality. 
\citet{babaioff2021wmms} proposed and studied additional extensions of the maximin share to the weighted setting. 

Several works study computational aspects in the weighted domain, proposing algorithms for weighted envy-freeness \cite{chakraborty2021weighted,chakraborty2021pickingsequences}, weighted Nash welfare \cite{garg2021mwnw,Suksumpong2022weightednash}, proportionality in chores \cite{li2022weightedproportional}, and weighted maximin share for chores \cite{aziz2019wmms}. 
\citeauthor{Suksumpong2022weightednash}'s approach uses path transfer arguments as well; however, our analysis holds for a more general class of valuations (binary submodular vs. binary additive).

The course allocation domain is a relevant application domain for our work. The main concern in applying our algorithmic framework in that domain is that course conflicts do not generally induce submodular preferences. Indeed, when items have arbitrary conflicts (encoded by a conflict graph), the problem of computing an allocation maximizing the minimal utility of any agent is computationally intractable \cite{chiarelli2022fairnessconflicts}. However, we find that in practice (at least in the authors' college), course conflicts satisfy a minimal condition that induces submodular preferences: if course A conflicts with course B, and course B conflicts with course C, then A conflicts with C.  
\section{Preliminaries}
We use $[t]$ to denote the set $\{1, 2, \dots, t\}$. 
For ease of readability, for an element $g$ and a set $A$, we replace $A \cup \{g\}$ and $A \setminus \{g\}$ with $A + g$ and $A - g$ respectively. 

We have a set of $n$ {\em agents} $N = [n]$ and a set of $m$ {\em goods} $G = \{g_1, g_2, \dots, g_m\}$. 
Each agent $i \in N$ has a {\em valuation function} $v_i: 2^G \rightarrow \R_{\ge 0}$; $v_i(S)$ specifies the value agent $i$ has for the set of goods $S \subseteq G$. Each agent $i$ also has a positive {\em weight} $w_i \in \R^+$ that corresponds to their {\em entitlement}; we do not place any other constraint on the entitlement other than positivity. We use $\Delta_i(S, g) = v_i(S+g) - v_i(S)$ to denote the marginal gain of adding the good $g$ to the bundle $S$ for the agent $i$.
Throughout the paper, we assume each $v_i$ is a {\em matroid rank function} (MRF). A function $v_i$ is a matroid rank function if 
\begin{inparaenum}[(a)]
    \item $v_i(\emptyset) = 0$, 
    \item for any $g \in G$ and $S \subseteq G$, we have $\Delta_i(S, g) \in \{0, 1\}$, and 
    \item for any $S \subseteq T \subseteq G$ and a good $g$, we have $\Delta_i(S, g) \ge \Delta_i(T, g)$.
\end{inparaenum}
These functions are also referred to as {\em binary submodular valuations}; we use the two terms interchangably.

An {\em allocation} $X$ is a partition of the set of goods into $n+1$ sets $(X_0, X_1, \dots, X_n)$ where each agent $i \in N$ gets the bundle $X_i$ and $X_0$ consists of the unallocated goods. An allocation $X$ is said to be {\em non-redundant} if for all $i \in N$, we have $v_i(X_i) = |X_i|$. For any allocation $X$, $v_i(X_i)$ is referred to as the {\em utility} or {\em value} of agent $i$ under the allocation $X$. For ease of analysis, we sometimes treat $0$ as an agent with valuation function $v_0(S) = |S|$ and bundle $X_0$. None of our justice criteria take the agent $0$ into account. By our choice of $v_0$, any allocation which is non-redundant for the set of agents $N$ is trivially non-redundant for the set of agents $N+0$. 

We have the following simple useful result about non-redundant allocations. Variants of this result have been shown by \citet{benabbou2021MRF} and \citet{viswanathan2022yankee}.

\begin{lemma}[\citet{benabbou2021MRF}, \citet{viswanathan2022yankee}]\label{lem:non-redundant-wlog}
Let $X$ be an allocation. There exists a non-redundant allocation $X'$ such that $v_i(X'_i) = v_i(X_i)$ for all $i \in N$.
\end{lemma}



\subsection{Item Exchange Graph}
We define the {\em exchange graph} of a non-redundant allocation $X$ (denoted by $\cal G(X)$) as a directed graph defined over the set of goods $G$. An edge exists from good $g \in X_i$ to another good $g'$ if $v_i(X_i - g + g') = v_i(X_i)$. Intuitively, this means that from the perspective of agent $i \in N+0$ (who owns $g$), $g$ can be replaced with $g'$ without reducing agent $i$'s utility. 
The exchange graph is a useful representation since it can be used to compute valid transfers of goods between agents. 

Let $P = (g_1, g_2, \dots, g_t)$ be a path in the exchange graph for the allocation $X$. 
We define a transfer of goods along the path $P$ in the allocation $X$ as the operation where $g_t$ is given to the agent who has $g_{t-1}$, $g_{t-1}$ is given to the agent who has $g_{t-2}$ and so on until finally $g_1$ is discarded. 
This transfer is called {\em path augmentation}; the bundle $X_i$ after path augmentation with the path $P$ is denoted by $X_i \Lambda P$ and defined as $X_i \Lambda P = (X_i - g_t) \oplus \{g_j, g_{j+1} : g_j \in X_i\}$ where $\oplus$ denotes the symmetric set difference operation. 
While the conventional notation for path augmentation uses $\Delta$ \citep{Barman2021MRFMaxmin,schrijver-book}, we replace it with $\Lambda$ to avoid confusion with the other definition of $\Delta$ as the marginal gain of adding an item to a bundle. 

For any non-redundant allocation $X$ and agent $i$, we define $F_i(X) = \{g \in G: \Delta_i(X_i, g) = 1\}$ as the set of goods which give agent $i$ a marginal gain of $1$. For any agent $i$, let $P = (g_1, \dots, g_t)$ be the shortest path from $F_i(X)$ to $X_j$ for some $j \ne i$. Then path augmentation with the path $P$ and giving $g_1$ to $i$ results in an allocation where $i$'s value for their bundle goes up by $1$, $j$'s value for their bundle goes down by $1$ and all the other agents do not see any change in value. This is formalized below and exists both in \citet[Lemma 5]{viswanathan2022yankee} and \citet[Lemma 1]{Barman2021MRFMaxmin}.

\begin{lemma}[\citet{viswanathan2022yankee, Barman2021MRFMaxmin}]\label{lem:path-augmentation}
    Let $X$ be a non-redundant allocation. Let $P = (g_1, \dots, g_t)$ be the shortest path from $F_i(X)$ to $X_j$ for some $i \in N+0$ and $j \in N + 0 - i$. We define an allocation $Y$ as follows:
    \begin{align*}
        Y_k = 
        \begin{cases}
            X_k \Lambda P & k \in N + 0 - i \\
            X_i \Lambda P + g_1 & k = i \\
        \end{cases}
    \end{align*}
    Then, we have for all $k \in N +0 - i - j$, $v_k(Y_k) = v_k(X_k)$, $v_i(Y_i) = v_i(X_i) + 1$ and $v_j(Y_j) = v_j(X_j) - 1$. Furthermore, the allocation $Y$ is non-redundant.
\end{lemma}
This lemma is particularly useful when the path ends at some good in $X_0$; transferring goods along the path results in an allocation where no agent loses any utility (they lose a good they like, but steal a good they like to recover their utility), but one agent (agent $i$ in \Cref{lem:path-augmentation}) increases their utility by $1$. 
We say there is a path from some agent $i$ to some agent $j$ in an allocation $X$ if there is a path from $F_i(X)$ to $X_j$ in the exchange graph $\cal G(X)$. \citet[Theorem 3.8]{viswanathan2022yankee} establish a sufficient condition for a path to exist; we present it below.

\begin{lemma}[\citet{viswanathan2022yankee}]\label{lem:augmentation-sufficient}
Let $X$ and $Y$ be two non-redundant allocations. 
For any agent $i \in N + 0$ such that $|X_i| < |Y_i|$, there exists a path in $\cal G(X)$ from $F_i(X)$ to $X_j$ for some $j \in N+0$ such that $|X_j| > |Y_j|$.
\end{lemma}

\subsection{Justice Criteria}
We define the {\em utility vector} of an allocation $X$ as $\vec u^X = (v_1(X_1), v_2(X_2), \dots, v_n(X_n))$. 
In general, a {\em justice criterion} (denoted by $\Psi$) is a way of comparing the utility vectors of two allocations $X$ and $Y$. We use $\vec u^X \succ_{\Psi} \vec u^Y$ to denote allocation $X$ being better than allocation $Y$ according to $\Psi$.
To ensure all allocations can be compared, we require $\succeq_{\Psi}$ be a total ordering on the set of all possible utility vectors $\mathbb{Z}^n_{\ge 0}$.

For example, if $\Psi$ is the Nash welfare justice criterion, 
$\vec u^X \succ_{\Psi} \vec u^Y$ if the product of agent utilities under $X$ is greater than the product of agent utilities under $Y$.
For readability, we abuse notation sometimes and replace $\vec u^X \succ_{\Psi} \vec u^Y$ with $X \succ_{\Psi} Y$.

Our goal is to compute an allocation with a \emph{maximal} utility vector with respect to $\Psi$. In other words, we would like to find an allocation $X$ such that for no other allocation $Y$, we have $Y \succ_{\Psi} X$. We sometimes refer to such an allocation as a $\Psi$ maximizing allocation.
In the above example, this would correspond to computing a max Nash welfare allocation.

\subsection{Important Definitions}
For ease of readability, we only define a few necessary terms from the fair division literature, and defer additional definitions to where they are used. 
\begin{definition}[Utilitarian Social Welfare]
The {\em utilitarian social welfare} of an allocation $X$ is given by $\sum_{i \in N} v_i(X_i)$. An allocation is referred to as \MAXUSW if it maximizes the utilitarian social welfare. 
\end{definition}
\begin{definition}[Lexicographic Domination]
Let $\vec x, \vec y \in \R^c$ be two vectors for some positive integer $c$. $\vec x$ is said to {\em lexicographically dominate} $\vec y$ if there exists a $k \in [c]$ such that for all $j \in [k-1]$, we have $\vec x_j = \vec y_j$ and we have $\vec x_k > \vec y_k$. A real valued vector $\vec x$ is {\em lexicographically dominating} with respect to a set of vectors $V$ if there exists no $\vec y \in V$ which lexicographically dominates $\vec x$.

This definition can be extended to allocations as well. An allocation $X$ is said to lexicographically dominate an allocation $Y$ if the utility vector of $X$ lexicographically dominates the utility vector of the $Y$. Similarly, an allocation $X$ is {\em lexicographically dominating} with respect to a set of allocations $\cal V$ if there exists no $Y \in \cal V$ which lexicographically dominates $X$.
\end{definition}
\begin{definition}[Pareto Dominance]
An allocation $X$ is said to {\em Pareto dominate} another allocation $Y$ if for all $h \in N$, we have $v_h(X_h) \ge v_h(Y_h)$ with a strict inequality for at least one $h \in N$. 
\end{definition}

\section{General Yankee Swap}
Our algorithmic fair allocation framework generalizes the Yankee Swap algorithm by \citet{viswanathan2022yankee}. 
In the Yankee Swap algorithm, all goods are initially unallocated. 
At every round, the agent with the least utility picks a good they like from the unallocated pile or initiate a {\em transfer path} where they steal a good they like from another agent, who then steals a good they like from another agent and so on until an agent finally takes a good they like from the set of unallocated goods. 
These transfer paths are equivalent to shortest paths on the exchange graph and can be computed easily.
If there is no such path to the pool of unallocated goods, the agent is removed from the game (denoted by their removal from the set $U$). 
We terminate once all agents are no longer playing.
 
We now present \emph{General Yankee Swap} (\Cref{algo:weighted-yankee-swap}). 
Surprisingly enough, we only change one line in the algorithm's pseudocode: instead of picking the least utility agent at every round, we pick an agent that maximizes a general \emph{gain function} $\phi$. The gain function $\phi$ takes as input the utility vector of an allocation (the partial allocation we have so far) and an index $i \in N$; its output is a $b$-dimensional vector. 
When $b > 1$, $\phi(\vec u^X, i) > \phi(\vec u^X, j)$ if $\phi(\vec u^X, i)$ lexicographically dominates $\phi(\vec u^X, j)$. 
If multiple agents maximize the gain function $\phi$, we break ties by choosing the agent with the least index.
For ease of readability, we sometimes replace $\phi(\vec u^X, i)$ with $\phi(X, i)$.

The gain function $\phi$ depends on the justice criterion we maximize (\Cref{table:results-summary} summarizes the main condition for the gain function, elaborated upon in \Cref{sec:applications}).
The original Yankee Swap is a specific case of the general Yankee Swap with $\phi(X, i) = -v_i(X_i)$.


\begin{algorithm}[t]
    \caption{General Yankee Swap}
    \label{algo:weighted-yankee-swap}
    \SetAlgoLined
    \DontPrintSemicolon
    $X = (X_0, X_1, \dots, X_n) \gets (G, \emptyset, \dots, \emptyset)$\; \tcp{All items initially in $X_0$, i.e. unassigned.}
    $U \gets N$\;
    \While{$U \ne \emptyset$}{
        $S \gets \argmax_{k \in U} \phi(X, k)$\; 
        \tcp{Choose the agent who maximizes $\phi$}
        $i \gets \min \{j : j \in S\}$\; 
        \tcp{Break ties using index}
        Find the shortest path in the exchange graph $\cal G(X)$ from $F_i(X)$ to $X_0$\;
        \uIf{a path $P = (g_{i_1}, g_{i_2}, \dots, g_{i_k})$ exists}{
            \tcp{Update $X$ using path augmentation}
            $X_k \gets X_k \Lambda P$ for all $k \in N - i$\; 
            $X_i \gets X_i \Lambda P + g_{i_1}$\;
            $X_0 \gets X_0 \Lambda P$\;
        }
        \Else{
            $U\gets U - i$\;
            \tcp{If no path exists, remove $i$ from $U$}
        }
    }
    \Return $X$\;
\end{algorithm}

\subsection{Sufficient Conditions for General Yankee Swap}
General Yankee Swap works when the justice criterion $\Psi$ has the following properties. These properties have been defined for arbitrary vectors $\vec x, \vec y$ and $\vec z$ but it may help to think of these vectors as utility vectors.
\begin{description}
    \item[(C1) --- Pareto Dominance:] For any two vectors $\vec x, \vec y \in \Z^n_{\ge 0}$, if $x_h \ge y_h$ for all $h \in N$, then $\vec x \succeq_{\Psi} \vec y$. Equality holds if and only if $\vec x = \vec y$.
    \item[(C2) --- Gain Function:] $\Psi$ admits a gain function $\phi$ that maps each possible utility vector to a real-valued $b$-dimensional vector with the following properties:
    \begin{description}
        \item[(G1)] For any vector $\vec x \in \Z^n_{\ge 0}$ and any $i, j \in [n]$, Let $\vec y \in \Z^n_{\ge 0}$ be the vector that results from starting at $\vec x$ and adding $1$ to $x_i$. 
        Similarly, let $\vec z \in \Z^n_{\ge 0}$ be the vector resulting from starting at $\vec x$ and adding $1$ to $x_j$. Then, if $\phi(\vec x, i) \ge \phi(\vec x, j)$, $\vec y \succeq_{\Psi} \vec z$. Equality holds if and only if $\phi(\vec x, i) = \phi(\vec x, j)$.
        \item[(G2)] For any two vectors $\vec x, \vec y \in \Z^n_{\ge 0}$ and $i \in [n]$, if $x_i \le y_i$, then $\phi(\vec x, i) \ge \phi(\vec y, i)$ with equality holding if $x_i = y_i$.
    \end{description}
\end{description}
Intuitively, $\phi(X, i)$ can be thought of as a function describing the marginal `gain' of giving a good to agent $i$ given some allocation $X$. The higher the value of $\phi(X, i)$, the more valuable it is to give an item to $i$. 
The condition (G1) states that if giving the item $g$ to agent $i$ results in a better allocation according to the justice criterion $\Psi$, then this should be reflected in the gain function $\phi$. The condition (G2) states that the gain function should take an egalitarian approach, assigning a greater $\phi$ value to $i$ when their utility is lower. 


\subsection{Analysis}
General Yankee Swap outputs a non-redundant $\Psi$ maximizing allocation which is also utilitarian welfare maximizing. Moreover, among all allocations that maximize $\Psi$, the output of General Yankee Swap is lexicographically dominating. 
This is a stronger statement and is required to show strategyproofness --- a technique used by \citet{halpern2020binaryadditive} and \citet{Suksumpong2022weightednash}.
Our proof is subtly different from that of \citet{viswanathan2022yankee}; while they heavily rely on the output allocation being leximin, our proof carefully uses path augmentations in the analysis to show correctness for any valid $\Psi$.
We first argue that the output of \Cref{algo:weighted-yankee-swap} is non-redundant, i.e. $v_i(X_i) = |X_i|$ for all $i \in N$.
\begin{lemma}\label{lem:general-yankee-nonredundant}
At any iteration of
\Cref{algo:weighted-yankee-swap}, the allocation $X$ maintained by the algorithm is non-redundant, i.e. $v_i(X_i) = |X_i|$ for all $i \in N$.
\end{lemma}
\begin{proof}
This claim holds via an inductive argument on the iterations of \Cref{algo:weighted-yankee-swap}. Let $X^t$ be the allocation at the $t$-th iteration. At iteration $t = 0$, no agent holds any item, thus the allocation is trivially non-redundant. 
At every other iteration, the only way we modify the allocation is using path augmentations. 
From Lemma \ref{lem:path-augmentation}, the new allocation after path augmentation is non-redundant as well. 
Therefore, $X^t$ must be non-redundant for any $t$.
\end{proof}

Let us next prove our main claim. 
\begin{theorem}\label{thm:weighted-yankee-leximin}
Let $\Psi$ be a justice criterion that satisfies (C1) and (C2) with a gain function $\phi$. 
When agents have matroid rank valuations, General Yankee Swap with input $\phi$ maximizes $\Psi$. 
Moreover, the output of General Yankee Swap lexicographically dominates all other $\Psi$-maximizing allocations.
\end{theorem}
\begin{proof}
It is easy to show that Algorithm \ref{algo:weighted-yankee-swap} always terminates: at every iteration, we either reduce the number of unallocated goods or remove some agent from $U$ while not changing the number of unallocated goods. 
Let $X$ be the non-redundant allocation output by Algorithm \ref{algo:weighted-yankee-swap}. 

Let $Y$ be a non-redundant allocation that maximizes $\Psi$ --- such an allocation is guaranteed to exist since $\Psi$ is a total order over all possible utility vectors. We can assume $Y$ is non-redundant thanks to Lemma \ref{lem:non-redundant-wlog}. 
If there are multiple such $Y$, pick one that lexicographically dominates all others --- that is, pick one that breaks ties in favor of lower index agents. 

If for all $h \in N$, $|X_h| \ge |Y_h|$, then $X$ maximizes $\Psi$ (since $\Psi$ respects Pareto dominance according to C1) and is lexicographically dominating --- we are done. Assume for contradiction that this does not hold.

This means that there is some agent $i$ whose utility under $X$ is strictly lower than under $Y$, i.e. $|X_i|< |Y_i|$. 
Let $i \in N$ be the agent with highest $\phi(X, i)$ in $X$ such that $|X_i| < |Y_i|$; if there are multiple agents we break ties in favor of the lowest index agent.
Let $W$ be the non-redundant allocation maintained by General Yankee Swap at the start of the iteration where $i$ was removed from $U$. From this stage onward, $i$'s utility will no longer increase. In addition, in order to be removed, $i$ must have the highest value of $\phi(W,i)$ among all agents in $U$.   
We use $t$ to denote this iteration. We have the following lemma. 
\begin{lemma}\label{lem:leximin-dominance}
For all $h \in N$, $|Y_h| \ge |W_h|$.
\end{lemma}
\begin{proof}
Assume for contradiction that this is not true.
Let $j \in N$ be the agent with highest $\phi(Y, j)$ such that $|Y_j| < |W_j|$; if there are multiple, break ties in favor of the one with the least index.

Consider the bundle $W_j$. 
Let $W'$ be the allocation at the start of the iteration when $j$ moved from a bundle of size $|W_j| - 1$ to $|W_j|$. That is, 
$j$ was the agent with maximum $\phi(W', j)$ within $U$, and executed a transfer path which resulted in them receiving an additional item with a marginal gain of $1$. 
Combining this with the fact that bundle sizes for each agent monotonically increase and that agents never return to $U$ once removed, we have 
$\phi(W',j) \ge \phi(W',i)$. Furthermore, since $\phi$ decreases with agent utilities (G2), $\phi(Y,j) \ge \phi(W',j)$ and $\phi(W',i)\ge \phi(W,i)$. Therefore:
$$\phi(Y, j) \ge \phi(W', j) \ge \phi(W', i) \ge \phi(W, i).$$ 
If equality holds throughout and $\phi(Y, j) = \phi(W', j) = \phi(W', i) = \phi(W, i)$, then $j < i$ since General Yankee Swap breaks final ties using the index of the agent and $j$ was chosen at $W'$ instead of $i$.
Therefore, we have 
\begin{align}
\phi(Y, j) \ge \phi(W, i) \text{. If equality holds, then }j < i. \label{obs:yjlesswj}
\end{align}

Invoking Lemma \ref{lem:augmentation-sufficient} with allocations $Y$ and $W$ and agent $j$, we can improve $j$'s utility under $Y$ via some transfer path that ends in some agent $k \in N+0$ for whom $v_k(Y_k)> v_k(W_k)$; that is, there must be a path from the items that offer $j$ a marginal gain of $1$ under $Y$ (the set $F_j(Y)$) to $Y_k$ for some $k \in N+0$ where $|Y_k| > |W_k|$ in the exchange graph of $Y$.  
Invoking Lemma \ref{lem:path-augmentation}, transferring goods along the shortest path from $F_j(Y)$ to $Y_k$ results in a non-redundant allocation $Z$ where $|Z_j| = |Y_j| + 1$, $|Z_k| = |Y_k| - 1$, and all other agents' utilities are the same.

If $k = 0$ --- i.e. $j$ executed a transfer path that ended with an unassigned item --- we are done since $Z$ strictly Pareto dominates $Y$. Using (C1), $Z \succ_{\Psi} Y$, contradicting our assumption on the $\Psi$ optimality of $Y$. Therefore, it must be that $k\ne 0$.

Consider $\phi(W, k)$. 
If $\phi(W, k) > \phi(W, i)$, since $i$ was chosen as the agent with highest $\phi(W, i)$ among the agents in $U$ at iteration $t$, we must have that $k\notin U$ at iteration $t$. 
This gives us $|X_k| = |W_k| < |Y_k|$. 
Combining this with our initial assumption that $\phi(W, k) > \phi(W, i)$, we get $\phi(X, k) =  \phi(W, k) > \phi(W, i) \ge \phi(X, i)$ using (G2). 
This contradicts our choice of $i$; $i$ is not the agent with highest $\phi(X, i)$ such that $|X_i| < |Y_i|$. 
This leads us to the following observation (combined with \Cref{obs:yjlesswj}).
\begin{align}
\phi(W, k) \le \phi(W, i) \le \phi(Y, j) \label{obs:kgreaterij}
\end{align}

Let $Y'$ be a non-redundant allocation that results from starting at $Y$ and removing any good from $k$. Note that, to show $Z \succeq_{\Psi} Y$, it suffices to show that $\phi(Y', j) \ge \phi(Y', k)$ (using (G1)).
From \Cref{obs:kgreaterij} and (G2), we have $\phi(Y', k) \le \phi(W, k) \le \phi(W, i) \le \phi(Y, j) = \phi(Y', j)$. 
If any of these inequalities are strict, we have $\phi(Y', j) > \phi(Y', k)$ which implies $Z \succ_{\Psi} Y$, a contradiction since $Y$ is $\Psi$-optimal.
This gives us the following observation. 
\begin{align}
    \phi(Y', k) = \phi(W, k) = \phi(W, i) = \phi(Y, j) = \phi(Y', j) \label{obs:long-equality}
\end{align}
Since $\phi(W, k) = \phi(W, i)$ and the algorithm picked $i$ at iteration $t$, we must have $i \le k$. 
Assume for contradiction that this is not true. 
If $k\in U$ at iteration $t$, the algorithm would have picked $k$ instead of $i$ --- a contradiction. 
If $k\notin U$, then $\phi(X, k) = \phi(W, k) = \phi(W, i) \ge \phi(X, i)$ using (G2).
Combining this with $i > k$ contradicts our choice of $i$. Therefore, $i \le k$. Combined with \Cref{obs:yjlesswj}, we have:
\begin{align}
j < i \le k \label{obs:equal-weights}
\end{align}
Combining \Cref{obs:long-equality} and (G1), we get that $Z$ maximizes $\Psi$.
However, since $j < k$, $Z$ lexicographically dominates $Y$: 
all agents $h < j$ receive the same value in both allocations and $v_j(Z_j) > v_j(Y_j)$. 
This contradicts our assumption on $Y$.
\end{proof}

We have $|W_i| < |Y_i|$ by construction and $|W_h| \le |Y_h|$ for all $h \in N$ thanks to Lemma \ref{lem:leximin-dominance}. Therefore, from Lemma \ref{lem:augmentation-sufficient}, there must exist a path from $F_i(W)$ to $W_0$ in the exchange graph of $W$. This is a contradiction since we chose the iteration where $i$ was removed from $U$ implying that there is no path from $F_i(W)$ to $W_0$ in the exchange graph of $W$. 
\end{proof}

Next, we show that the output of the General Yankee Swap is always \MAXUSW.

\begin{prop}\label{prop:yankee-swap-usw}
For any input gain function $\phi$, the output of General Yankee Swap is \MAXUSW.
\end{prop}
\begin{proof}
Let $X$ be the non-redundant allocation output by General Yankee Swap. Assume for contradiction that $X$ is not \MAXUSW.
Let $Y$ be a non-redundant \MAXUSW allocation which minimizes $\sum_{h \in N}|v_h(X_h) - v_h(Y_h)|$. 

There must be at least one agent $i$ such that $|X_i| < |Y_i|$. 
Consider the allocation $W$ at the start of the iteration where $i$ was removed from $U$. Since $|W_i| < |Y_i|$ and $|W_h| \le |Y_h|$ for all $h \in N$, using Lemma \ref{lem:augmentation-sufficient}, there exists a path from $F_i(W)$ to $W_0$ in $\cal G(W)$. This contradicts our choice of iteration since we chose the iteration where $i$ was removed from $U$.

Therefore, we must have at least one agent $j$ such that $|Y_j| < |X_j|$. Applying Lemma \ref{lem:augmentation-sufficient} with allocations $X$, $Y$ and the agent $j$, we get that there is a path from $j$ to some agent $k \in N+0$ in the exchange graph of $Y$ such that $|X_k| < |Y_k|$. Transferring goods along the shortest path from $j$ to $k$, using Lemma \ref{lem:path-augmentation}, leads to a non-redundant allocation $Z$ where $|Z_j| = |Y_j| + 1$ and $|Z_k| = |Y_k| - 1$ and all other agents receive the same utility. 
If $k = 0$, $Z$ has a higher USW than $Y$ contradicting our assumption on $Y$.

If $k \ne 0$, then $\sum_{h \in N}|v_h(X_h) - v_h(Y_h)| > \sum_{h \in N}|v_h(X_h) - v_h(Z_h)|$ and $Z$ is \MAXUSW; again contradicting our assumption on $Y$.
Therefore, $X$ is \MAXUSW.
\end{proof}

\section{Strategyproofness}\label{sec:strategyproofness}
In this section, we show that if preferences are elicited {\em before} running the General Yankee Swap, being truthful is the dominant strategy. 

A mechanism is said to be {\em strategyproof} if no agent can get a better outcome by misreporting their valuation function. 
We define the steps of the Yankee Swap Mechanism as follows:
\begin{enumerate}
    \item Elicit the valuation function $v_i$ of each agent $i \in N$. 
    If an agent's valuation function is not an MRF, set the agent's valuation of every bundle to $0$.
    \item Use General Yankee Swap to compute a non-redundant allocation that maximizes some valid $\Psi$ for the valuation profile $\{v_i\}_{i \in N}$.
\end{enumerate}
Before we show the final result, we prove some useful lemmas. Our proof uses the same ideas as the strategyproofness result in \citet[Theorem 5]{Babaioff2021Dichotomous}.
Given a set $T \subseteq G$, we define the function $f_T: 2^G \rightarrow \R^+$ as $f_T(S) = |S \cap T|$. Note that for any $T$, $f_T$ is an MRF.
\begin{lemma}\label{lem:faithfulness}
Let $X$ be the output allocation of the Yankee Swap mechanism with valid input gain function $\phi$ and valuation profile $\{v_i\}_{i \in N}$. 
For some agent $i \in N$, replace $v_i$ with some $f_T$ such that $T \subseteq X_i$ and run the mechanism again to get an allocation $Y$. We must have $Y_i = T$.
\end{lemma}
\begin{proof}
Since the allocation $Y$ is non-redundant, we have that $Y_i \subseteq T$. Assume for contradiction that $Y_i \ne T$. Define an allocation $Z$ as $Z_h = X_h$ for all $h \in N - i$ and $Z_i = T$; allocate the remaining goods in $Z$ to $Z_0$. Note that both $Y$ and $Z$ are non-redundant under both valuation profiles (with the old $v_i$ and the new valuation function $f_T$).

Let us compare $Y$ and $Z$. 
Let $p \in N$ be the agent with highest $\phi(Y, p)$ such that $|Y_p| < |Z_p|$; choose the agent with the least $p$ if there are ties.  Such an element is guaranteed to exist since $|Y_i| < |Z_i|$.

If there exists no $q \in N$ such that $|Y_q| > |Z_q|$, we must have $|Y_0| > |Z_0|$ (since $|Y_i| < |Z_i|$). 
Using Lemma \ref{lem:augmentation-sufficient} with agent $p$ and the allocations $Y$ and $Z$, we get that there is a path from $p$ to $0$ in the exchange graph of $Y$. 
Transferring goods along the shortest such path results in an allocation with a higher USW than $Y$ under the new valuation profile contradicting the fact that $Y$ is \MAXUSW.

Let $q \in N$ be the agent with highest $\phi(Z, q)$ such that $|Y_q| > |Z_q|$; break ties by choosing the least $q$. Further, note that since $Z$ and $X$ only differ in $i$'s bundle and $|Y_i| < |Z_i|$, we must have $|X_q| = |Z_q|$.

Consider two cases:
\begin{enumerate}[(i)]
    \item $\phi(X, q) > \phi(Y, p)$,
    \item $\phi(X, q) = \phi(Y, p)$ and $q < p$
\end{enumerate}
Then invoking Lemma \ref{lem:augmentation-sufficient} with allocations $X$, $Y$ and the agent $q$, there exists a transfer path from $q$ to some agent $k$ in the exchange graph of $X$ (w.r.t. the old valuations) where $|Y_k| < |X_k|$. Transferring along the shortest such path gives us a non-redundant allocation $X'$ where $|X'_q| =|X_q| + 1$ and $|X'_k| = |X_k| - 1$ (Lemma \ref{lem:path-augmentation}). 
Let $X''$ be an allocation starting at $X'$ and removing one good from $X'_q$.
If $\phi(X'', q) > \phi(X'', k)$, then $X' \succ_{\Psi} X$ (using (G1)) contradicting our assumption on $X$.

If $k = 0$, we improve USW contradicting the fact that $X$ is \MAXUSW with respect to the original valuations $\{v_h\}_{h \in N}$. 

\noindent\textbf{For case (i):} If $k \ne 0$, we have $\phi(X'', q) = \phi(X, q) > \phi(Y, p) \ge \phi(Y, k) \ge \phi(X'', k)$ (using (G1) and (G2)). Therefore $\phi(X'', q) > \phi(X'', k)$ and $X$ does not maximize $\Psi$ --- a contradiction.

\noindent\textbf{For case (ii):} If $k \ne 0$, we have $\phi(X'', q) = \phi(X, q) = \phi(Y, p) \ge \phi(Y, k) \ge \phi(X'', k)$. If any of these weak inequalities are strict, we can use analysis similar to that of case (i) to show that $X$ does not maximize $\Psi$. 
Therefore, all the weak inequalities must be equalities and we must have $\phi(X'', q) = \phi(X, q) = \phi(Y, p) = \phi(Y, k) = \phi(X'', k)$. This implies that $X =_{\Psi} X'$ using (G1).

Moreover, by our choice of $p$ we have $p \le k$ and by assumption, we have $q < p$. Combining the two, this gives us $q < k$. Therefore, $X'$ lexicographically dominates $X$ --- a contradiction to Theorem \ref{thm:weighted-yankee-leximin}.

Let us now move on to the remaining two possible cases
\begin{enumerate}[(i)]\addtocounter{enumi}{2}
    \item $\phi(Z, q) < \phi(Y, p)$,
    \item $\phi(Z, q) = \phi(Y, p)$ and $q > p$
\end{enumerate}
Recall that both $Y$ and $Z$ are non-redundant with respect to the new valuation profile (with $f_T$). If any of the above two conditions occur, then invoking Lemma \ref{lem:augmentation-sufficient} with allocations $Y$, $Z$ and the agent $p$, there exists a transfer path from $p$ to some agent $l$ in the exchange graph of $Y$ where $|Y_l| > |Z_l|$. Transferring along the shortest such path gives us a non-redundant allocation $Y'$ where $|Y'_p| = |Y_p| + 1$ and $|Y'_l| = |Y_l| - 1$ (Lemma \ref{lem:path-augmentation}). 

Let $Y''$ be an allocation starting at $Y'$ and removing one good from $Y'_p$.
If $\phi(Y'', p) > \phi(Y'', l)$, then $Y' \succ_{\Psi} Y$ contradicting our assumption on $Y$.

If $l = 0$, we improve USW contradicting the fact that $Y$ is \MAXUSW with respect to the new valuation $f_T$. 

\noindent\textbf{For case (iii):} If $l \ne 0$, we have $\phi(Y'', p) = \phi(Y, p) > \phi(Z, q) \ge \phi(Z, l) \ge \phi(Y'', l)$ (using (G1) and (G2)). Therefore $\phi(Y'', p) > \phi(Y'', l)$ and $Y$ does not maximize $\Psi$ --- a contradiction.

\noindent\textbf{For case (iv):} If $l \ne 0$, we have $\phi(Y'', p) = \phi(Y, p) = \phi(Z, q) \ge \phi(Z, l) \ge \phi(Y'', l)$.  If any of the weak inequalities are strict, we can use analysis similar to that of case (iii) to show that $Y$ does not maximize $\Psi$. 

Therefore, all the weak inequalities must be equalities and we must have $\phi(Y'', p) = \phi(Y, p) = \phi(Z, q) = \phi(Z, l) = \phi(Y'', l)$. This implies that $Y' =_{\Psi} Y$.

Moreover, by our choice of $q$ we have $q \le l$ and by assumption, we have $p < q$. Combining the two, this gives us $p < l$. Therefore, $Y'$ lexicographically dominates $Y$ --- contradicting the fact that $Y$ is a $\Psi$-maximizing allocation with respect to the new valuation profile (with $f_T$).

Since $|Z_q| = |X_q|$, $\phi(Z, q) = \phi(X, q)$. Therefore, cases (i)--(iv) cover all possible cases. Each of the cases lead to a contradiction. Therefore, our proof is complete and $Y_i = T$. 
\end{proof}
\begin{restatable}{lemma}{lemmonotonicity}\label{lem:monotonicity}
Let $X$ be the output allocation of the Yankee Swap mechanism with input valuation profile $\{v_i\}_{i \in N}$. For some agent $i \in N$, replace $v_i$ with some $v'_i$ such that $v'_i(S) \ge v_i(S)$ for all $S \subseteq G$ and run the mechanism again to get an allocation $Y$. We must have $|Y_i| \ge |X_i|$.
\end{restatable}
\begin{proof}[Proof Sketch]
Assume for contradiction that $|Y_i| < |X_i|$. Let $T$ be a subset of $Y_i$ such that $|T| = v_i(T) = v_i(Y_i)$. Define an allocation $Z$ as $Z_h = Y_h$ for all $h \in N - i$ and $Z_i = T$; allocate the remaining goods in $Z$ to $Z_0$. $X$ and $Z$ are non-redundant under both valuation profiles. This construction allows us to use a similar case by case analysis to that of Lemma \ref{lem:faithfulness} to prove the required result.
\end{proof}
We are now ready to show strategyproofness. 
\begin{theorem}
When agents have matroid rank valuations, the Yankee Swap mechanism is strategyproof.
\end{theorem}
\begin{proof}
Assume an agent $i$ reports $v'_i$ instead of their true valuation $v_i$ to generate the allocation $X'$ via the Yankee Swap mechanism. Let $X$ be the allocation generated by the mechanism had they reported their true valuation $v_i$. We need to show that $v_i(X_i) \ge v_i(X'_i)$.
We can assume w.l.o.g. that $v'_i$ is an MRF; otherwise, $i$ gets nothing and $v_i(X'_i) = 0$.

Let $B$ be a subset of $X'_i$ such that $|B| = v_i(B) = v_i(X'_i)$. 
Using Lemma \ref{lem:faithfulness}, we get that replacing $v'_i$ with $f_B$ will result in an allocation $Y$ where $Y_i = B$. 
Using Lemma \ref{lem:monotonicity}, we get that replacing $f_B$ with $v_i$ gives us the allocation $X$ and the guarantee $v_i(X_i) = |X_i| \ge |B| = v_i(B)$. 
Note that we can apply Lemma \ref{lem:monotonicity} since by construction we have $v_i(S) \ge v_i(S \cap B) = |S \cap B| = f_B(S)$ for all $S$ such that $S \cap B \ne \emptyset$ and $v_i(S) \ge 0 = f_B(S)$ otherwise.
Since $v_i(B) = v_i(X'_i)$, the proof is complete.
\end{proof}

\section{Time Complexity}\label{sec:time-complexity}
We turn to analyzing the time complexity of General Yankee Swap. We assume that agent valuations are computed using an oracle which takes $T_v = \Omega(m)$ time; since reading the input bundle alone will take any oracle $\Omega(m)$ time. 
This assumption only exists to simplify the time complexity expression; it does not affect the number of valuation queries made by the algorithm in any way. Our result uses the technique of combining binary search with breadth first search first introduced by \citet{Chakrabarty2019MatroidIntersection}.

We store the allocation using two data structures: 
\begin{inparaenum}[(a)]
    \item a binary matrix referred to as $X$, and 
    \item an inverse mapping $X^{-1}$ that maps each good to the agent it is allocated to in $X$.
\end{inparaenum}

\begin{algorithm}[t]
    \caption{\FindDesired$(i, S, B)$}
    \label{algo:find-desired}
    \DontPrintSemicolon
    \SetKwInOut{Input}{Input}
    \SetKwInOut{Output}{Output}
    \Input{An agent $i$, a bundle $S \subseteq G$, and another bundle $B \subseteq G$}
    \Output{An element $g \in B$ such that $\Delta_i(S, g) = 1$ or $\emptyset$ if no such element exists}
    \If{$v_i(S\cup B) = v_i(S)$}{
        \Return $\emptyset$\;
    }
    \While{$|B| > 1$}{
        Let $B_1$ and $B_2$ be a partition of $B$ such that $\max\{|B_1|, |B_2|\} \le \big  \lceil \frac{|B|}{2} \big \rceil$\;
        \uIf{$v_i(S \cup B_1) > v_i(S)$}{
            $B \gets B_1$\;
        }
        \Else{
            $B \gets B_2$\;
        }
    }
    \Return $B$\;
\end{algorithm}

Most proofs in this section are straightforward and have been relegated to the appendix. We start with a simple binary search procedure to check if, given an agent $i$ and a bundle $S$, there is a good $g$ in some bundle $B \subseteq G$ such that $\Delta_i(S, g) = 1$. 
The algorithm is simple: we first check if $v_i(S \cup B) > v_i(S)$. 
If this condition is satisfied, we partition the set $B$ into two equal sized sets $B_1$ and $B_2$ and check if $v_i(S \cup B_j) > v_i(S)$ for each $j \in \{1, 2\}$. 
If the condition holds for any one $B_j$ (say $B_1$), we repeat the process by dividing $B_1$ into two subsets until $B_1$ is a singleton element. 
We refer to this procedure as \FindDesired{} (Algorithm \ref{algo:find-desired}).
\begin{restatable}{lemma}{lemfinddesired}\label{lem:find-desired}
The procedure $\FindDesired(i, S, B)$ runs in $O(T_v \log |B|)$ time and finds a good $g \in B$ such that $\Delta_i(S, g) = 1$ if it exists; otherwise, the procedure outputs $\emptyset$.
\end{restatable}
We can use the \FindDesired{} procedure to find shortest paths in the exchange graph using breadth first search without explicitly building the exchange graph. 
We refer to this procedure as \GetDistances{} and its steps can be found in Algorithm \ref{algo:get-distances}. 

To find the shortest path from $F_i(X)$ to $X_0$ in the exchange graph, we add a source node $s$ and edges from $s$ to all the goods in $F_i(X)$. We then use a slightly modified breadth first search to find single source shortest paths in the exchange graph from $s$; the slight modification being that we use \FindDesired{} to find outgoing edges. 

\begin{algorithm}[t]
    \caption{\GetDistances$(X, i)$}
    \label{algo:get-distances}
    \DontPrintSemicolon
    \SetKwInOut{Input}{Input}
    \SetKwInOut{Output}{Output}
    \Input{An agent $i$ chosen by General Yankee Swap and a non-redundant allocation $X$}
    \Output{Shortest paths from $s$ to every good $g$ in $\cal G(X)$}
    Let $d_g \gets \infty, \prev_g \gets \texttt{None}$ for all $g \in G$\;
    $Q \gets \{s\}$, $B \gets G$\;
    \While{$Q \ne \emptyset$}{
        Let $a$ be the element added to $Q$ the earliest\;
        \uIf{$a = s$}{
            \While{$b = \FindDesired(i, X_i, B)$ satisfies $b \ne \emptyset$}{
                $d_b \gets 1$, $\prev_b \gets s$, $Q \gets Q +b$, $B \gets B - b$\;
            }
            $Q \gets Q - s$\;
        }
        \Else{
            $j = X^{-1}(a)$\;
            \tcp{$a$ is allocated to $j$ under $X$}
            \While{$b = \FindDesired(j, X_j - a, B)$ satisfies $b \ne \emptyset$}{
                $d_b \gets d_a + 1$, $\prev_b \gets a$, $Q \gets Q +b$, $B \gets B - b$\;
            }
            $Q \gets Q - a$\;
        }
    }
    \Return $d, \prev$\;
\end{algorithm}

\begin{restatable}{lemma}{lemgetdistances}\label{lem:get-distances}
On input a non-redundant allocation $X$ and an agent $i$, the procedure $\GetDistances(X, i)$ runs in $O(m T_v \log m)$ time and computes the distances and the shortest paths from the node $s$ to each good $g \in G$.
\end{restatable}

We can put these two results together to get the time complexity of the algorithm. Our proof uses some observations from \citet{viswanathan2022yankee}.

\begin{theorem}\label{thm:general-yankee-time}
Assuming that the worst case time to compute the value of any bundle of goods is $\Omega(m)$, 
the General Yankee Swap algorithm runs in $O([m T_v\log m + n(b+T_\phi)](m+n))$ time; where $T_v$ is the complexity of computing the value of a bundle of goods and $T_\phi$ is the complexity of computing $\phi$.
\end{theorem}
\begin{proof}
We make three observations. First, the algorithm runs for at most  $(m+n)$ iterations. At each round either $|X_0|$ reduces by $1$ or an agent is removed from $U$. $X_0$ monotonically decreases and agents do not return to $U$; therefore we can only have at most $m+n$ iterations. 

Second, finding shortest paths in the exchange graph takes $O(m T_v \log m)$ time (Lemma \ref{lem:get-distances}) and updating the allocation using path augmentation takes $O(m)$ time.

Finally, finding $i$ involves computing $\phi$ for each agent and then comparing the $\phi$ values. Since the output of $\phi$ has $b$ components, each comparison takes at most $O(b)$ time. Therefore, finding $i$ takes $n(b + T_{\phi})$ time.
Combining these three observations, we get the required time complexity.
\end{proof}

Due to the similarity of our algorithm with that of \citet{viswanathan2022yankee}, their time complexity result applies to our algorithm as well. However, their analaysis uses a naive implementation of breadth first search that uses $O((m+n)m^2)$ valuation queries. Our implementation on the other hand improves this to $O((m+n)m\log m)$ valuation queries. 

\section{Applying General Yankee Swap}\label{sec:applications}
In this section, we show how Yankee Swap can be applied to optimize commonly used justice criteria. This section showcases how simple the problem of optimizing fairness objectives becomes when using General Yankee Swap.

\subsection{Prioritized Lorenz Dominating Allocations}
As a sanity check, we first show how General Yankee Swap computes prioritized Lorenz dominating allocations. 

An allocation $X$ is {\em Lorenz dominating} if for all allocations $Y$ and for any $k \in [n]$, the sum of the utilities of the $k$ agents with least utility in $X$ is at least as much as the sum of the utilities of the $k$ agents with least utility in $Y$. 
An allocation $X$ is {\em leximin} if it maximizes the lowest utility and subject to that; maximizes the second lowest utility and so on.

Both these metrics can be formalized using the sorted utility vector. The {\em sorted utility vector} of an allocation $X$ (denoted by $\vec s^X$) is defined as the utility vector $\vec u^X$ sorted in ascending order. 
An allocation $X$ is Lorenz dominating if for all allocations $Y$ and all $k \in [n]$, we have $\sum_{j \in [k]} s_j^X \ge \sum_{j \in [k]} s_j^Y$.
An allocation $X$ is leximin if the sorted utility vector of $X$ is not lexicographically dominated by the sorted utility of any other allocation. 
A Lorenz dominating allocation is not guaranteed to exist, but when it does, it is equivalent to a leximin allocation (which is guaranteed to exist). 
This result holds for arbitrary valuation functions.

\begin{lemma}\label{lem:lorenz-dominance-leximin}
When a Lorenz dominating allocation exists, an allocation is leximin if and only if it is Lorenz dominating.
\end{lemma}
\begin{proof}
Let $Y$ be any Lorenz dominating allocation and let $X$ be any leximin allocation. 
Assume for contradiction that they do not have the same sorted utility vector. Let $k$ be the lowest index such that $s^X_k \ne s^Y_k$. 
If $s^X_k < s^Y_k$, then $\vec s^Y$ lexicographically dominates $\vec s^X$ contradicting the fact that $X$ is leximin. If $s^X_k > s^Y_k$, then $Y$ is not Lorenz dominating. 
Since $s_k^X = s_k^Y$ for all $k$, both allocations have the same sorted utility vector. This implies that $X$ is Lorenz dominating and $Y$ is leximin.
\end{proof}

\citet{Babaioff2021Dichotomous} introduce and study the concept of {\em prioritized Lorenz dominating allocations}. Each agent is given a priority which is represented using a permutation $\pi: [n] \mapsto [n]$. 
When agents have MRF valuations $\{v_i\}_{i \in N}$, prioritized Lorenz dominating allocations are defined as Lorenz dominating allocations for the fair allocation instance where valuations are defined as $v'_i(S) = v_i(S) + \frac{\pi(i)}{n^2}$; we refer to $v'$ as {\em perturbed valuations}. 
\citet{Babaioff2021Dichotomous} show that when agents have MRF valuations, a prioritized Lorenz dominating allocation is guaranteed to exist and satisfies several desirable fairness properties such as leximin, envy freeness up to any good (EFX) and maximizing Nash welfare. Prioritized Lorenz dominating allocations can be computed using the following gain function:

\begin{theorem}\label{thm:lorenz-dominance}
When agents have MRF valuations, General Yankee Swap with $\phi(X, i)$ set to  $(-v_i(X_i), -\pi(i))$ computes prioritized Lorenz dominating allocations with respect to priority $\pi$.
\end{theorem}
\begin{proof}
Since {\em Lorenz domination} is not a total ordering over the set of possible utility vectors, we instead compute {\em leximin allocations} under the perturbed valuations $v'$. Such an allocation is guaranteed to be a prioritized Lorenz dominating allocation since prioritized Lorenz dominating allocations are guaranteed to exist \citep{Babaioff2021Dichotomous} and are equivalent to leximin allocations when they do exist (Lemma \ref{lem:lorenz-dominance-leximin}).

Formally, for any two allocations $X$ and $Y$, $X \succ_{\Psi} Y$ if $\vec s^X$ lexicographically dominates $\vec s^Y$ where the sorted utility vectors $\vec s^X$ and $ \vec s^Y$ are defined according to the perturbed valuations $v'$. Note crucially that agents still have MRF valuations; the perturbed valuations $v'$ are only used to define the justice criterion $\Psi$.
$\Psi$ trivially satisfies Pareto Dominance (C1) and $\phi$ trivially satisfies (G2). We therefore only show (G1).

For any vector $\vec x \in \Z^n_{\ge 0}$ and two agents $i$ and $j$, let $\vec y \in \Z^n_{\ge 0}$ be the vector resulting from starting at $\vec x$ and adding $1$ to $ x_i$. Similarly, let $\vec z \in \Z^n_{\ge 0}$ be the vector resulting from starting at $\vec x$ and adding a value of $1$ to $x_j$. We assume that $\phi(\vec x, i) > \phi(\vec x, j)$, and show that $\vec y \succ_{\Psi} \vec z$. 
Note that since $\pi(i) \ne \pi(j)$, $\phi(\vec x,i)$ can never equal $\phi(\vec x,j)$.


If $\phi(\vec x, i) > \phi(\vec x, j)$ one of the following two cases must be true.

\noindent\textbf{Case 1:} $x_i < x_j$. If this is true, we have $x_j + \frac{\pi(j)}{n^2} > x_i + \frac{\pi(i)}{n^2}$ since $\frac{\pi(i)}{n^2} - \frac{\pi(j)}{n^2} < 1 \le x_j - x_i$.
Therefore $\vec y \succ_{\Psi} \vec z$ since it is always better to add utility to a lower valued agent (according to $v'$).

\noindent\textbf{Case 2:} $x_i = x_j$ and $\pi(i) < \pi(j)$. If this is true, we have $x_j + \frac{\pi(j)}{n^2} > x_i + \frac{\pi(i)}{n^2}$ by assumption. Again, $\vec y \succ_{\Psi} \vec z$ since it is always better to add utility to a lower valued agent (according to $v'$).
\end{proof}

\subsection{Weighted Leximin Allocations}
Let us next consider the case where agents have entitlements.
When each agent $i$ has a positive weight $w_i$, the weighted utility of an agent $i$ is defined as $\frac{v_i(X_i)}{w_i}$. A weighted leximin allocation maximizes the least weighted utility and subject to that, maximizes the second least weighted utility and so on. 

More formally, we define the {\em weighted sorted utility vector} of an allocation $X$ (denoted by $\vec e^X$) as $\left(\frac{v_1(X_1)}{w_1}, \frac{v_2(X_2)}{w_2}, \dots, \frac{v_n(X_n)}{w_n}\right)$ sorted in ascending order. 
An allocation $X$ is weighted leximin if for no other allocation $Y$, $\vec e^Y$ lexicographically dominates $\vec e^X$. We have the following result, the proof of which has been relegated to the appendix due to its similarity with Theorem \ref{thm:lorenz-dominance}.

\begin{restatable}{theorem}{thmweightedleximin}\label{thm:weighted-leximin}
When agents have MRF valuations and each agent $i$ has a weight $w_i$, General Yankee Swap with $\phi(X, i) = (-\frac{v_i(X_i)}{w_i}, -w_i)$ computes a weighted leximin allocation.
\end{restatable}
As \citet{chakraborty2021pickingsequences} show, the weighted leximin solution behaves in counterintuitive ways. 
Consider the case of a single item $g$ worth $1$ to two agents. Agent 1 has a weight of $w_1=1$ and agent 2 has a weight of $w_2=2$. Giving the item to the higher priority agent (agent 2) results in the sorted utility vector $(0,\frac{v_2(g)}{w_2}) = (0,\frac12)$, whereas giving the item to agent 1 results in the vector $(0,\frac{v_1(g)}{w_1}) = (0,1)$, which is lexicographically dominant. In other words, giving items to lower priority agents is \emph{better}. 
However, this undesirable behavior only occurs when there are fewer items than agents. 
In the run of General Yankee Swap, once every agent receives one item, the gain function prioritizes higher weight agents as should be expected.      
Thus, we believe that weighted leximin allocations can still be reasonably considered in the weighted domain. 
\subsection{Individual Fair Share Allocations}\label{sec:MMS}
We now turn to justice criteria that guarantee each agent a minimum \emph{fair share} amount. One such popular notion is the {\em maximin share}. 
An agent's maximin share \citep{Budish2011EF1} is defined as the utility an agent would receive if they divided the set of goods into $n$ bundles themselves and picked the worst bundle. 
More formally, the maximin share of an agent $i$ (denoted by $\MMS_i$) is defined as $\MMS_i = \max_{X = (X_1, X_2, \dots, X_n)} \min_{j \in [n]} v_i(X_j)$. 
There are several other fair share metrics popular in the literature \citep{farhadi2019wmms, babaioff2021wmms, babaioff2022fairshare}. 
All of these metrics have the same objective --- each agent $i$ has an instance dependent {\em fair share} $c_i \ge 0$; 
the goal is to compute allocations that guarantee each agent a high fraction of their share \citep{procaccia2014fairenough, ghodsi2018fair}.

We define the {\em fair share fraction} of an agent $i$ in an allocation $X$ as $\frac{v_i(X_i)}{c_i}$ when $c_i > 0$ and $0$ when $c_i = 0$. 
When $c_i = 0$, any bundle of goods (even the empty bundle) provides $i$ their fair share; 
therefore, an agent $i$ with $c_i = 0$ can be ignored when allocating bundles. 
When agents have matroid rank valuations, General Yankee Swap can be used to maximize the lowest fair share fraction received by an agent and subject to that, maximize the second lowest fair share fraction and so on. 
Using a proof very similar to Theorem \ref{thm:weighted-leximin}, the appropriate $\phi(X, i)$ to achieve such a fairness objective is defined as follows:
\begin{align}
    \phi(X, i) = 
    \begin{cases}
        (-\frac{v_i(X_i)}{c_i}, -c_i) & c_i > 0 \\
        (-M, 0) & c_i = 0
    \end{cases}\label{eq:fair-share}
\end{align}
where $M$ is a large positive number greater than any possible $\frac{v_i(X_i)}{c_i}$.
This can be seen as setting the weight of each agent $i$ to their share $c_i$ and computing a weighted leximin allocation. 
The only minor change we make is accounting for cases where $c_i = 0$: in such a case we make $\phi(X, i)$ the lowest possible value it can take so we do not allocate any goods to these agents. 
The only time Yankee Swap allocates items to these agents is when all the other agents with positive shares do not derive a positive marginal gain from any of the remaining unallocated goods.
More formally, given an allocation $Y$, let $\vec \mu(Y)$ be the \emph{sorted fair-share normalized utilty vector} of the values $\frac{v_i(Y_i)}{c_i}$ sorted in increasing order for all agents with $c_i > 0$ (we ignore the agents whose fair share is $0$), then: 
\begin{theorem}\label{thm:fair-share}
Let $X$ be the allocation output by General Yankee Swap with the gain function $\phi$ defined in \Cref{eq:fair-share}. Then, among all possible allocations, $X$ has a lexicographically dominating sorted fair-share normalized utility vector $\vec \mu(X)$. 
\end{theorem}
The proof of \Cref{thm:fair-share} is very similar to that of \Cref{thm:weighted-leximin}, and is thus omitted.
As an immediate corollary of Theorem \ref{thm:fair-share}, it is no longer necessary to find a fair share fraction that can be guaranteed to all agents, and then design an algorithm which allocates each agent at least this fraction of their fair share. 
General Yankee Swap automatically computes an allocation which maximizes the lowest fair share fraction received by any agent.
A straightforward corollary is that, when there exists an allocation that guarantees each agent their fair share, Yankee Swap outputs one such allocation. \citet{Barman2021MRFMaxmin} show that when agents have MRF valuations, an allocation which guarantees each agent their maximin share always exists. 
Using their result with Theorem \ref{thm:fair-share}, we have the following Corollary.

\begin{corr}\label{corr:MMS-MRF}
When agents have MRF valuations and every agent has $c_i = \MMS_i$, General Yankee Swap run with $\phi(X, i)$ given by \eqref{eq:fair-share} computes a \MAXUSW and \MMS allocation.
\end{corr}

\citet{Barman2021MRFMaxmin} also present a polynomial time algorithm to compute the maximin share of each agent. 
Since the procedure to compute maximin shares is not necessarily strategyproof, the overall procedure of computing a maximin share allocation using Yankee Swap may not be strategyproof either.

\subsection{Max Weighted Nash Welfare Allocations}
We still assume that each agent $i$ has a weight $w_i > 0$.
For any allocation $X$, let $P_X$ be the set of agents who receive a positive utility under $X$. 
An allocation $X$ is said to be {\em max weighted Nash welfare} (denoted by \MWNW) if it first minimizes the number of agents who receive a utility of zero; subject to this, $X$ maximizes $\prod_{i \in P_X} v_i(X_i)^{w_i}$ \citep{chakraborty2021weighted, Suksumpong2022weightednash}. 
We define the gain function $\phi(X, i)$ as follows:
\begin{align}
    \phi(X, i) = 
    \begin{cases}
        \big (1 + \frac{1}{v_i(X_i)} \big )^{w_i} & v_i(X_i) > 0\\
        M & v_i(X_i) = 0
    \end{cases} \label{eq:phi-mwnw}
\end{align}
where $M$ is a large number greater than any possible $\big (1 + \frac{1}{v_i(X_i)} \big )^{w_i}$. 
We have the following Theorem.
\begin{theorem}\label{thm:weighted-nash}
When agents have MRF valuations and each agent $i$ has a weight $w_i$, General Yankee Swap with $\phi$ given by \eqref{eq:phi-mwnw} computes a max weighted Nash welfare allocation.
\end{theorem}
\begin{proof}
Formally, we have for two allocations $X$ and $Y$, $X \succ_{\Psi} Y$ if any of the following two conditions hold
\begin{enumerate}[(a)]
    \item $|P_X| > |P_Y|$
    \item $|P_X| = |P_Y|$ and $\prod_{i \in P_X} v_i(X_i)^{w_i} > \prod_{i \in P_Y} v_i(Y_i)^{w_i}$
\end{enumerate}

It is easy to see that $\Psi$ satisfies Pareto dominance (C1) and $\phi$ satisfies (G2). 
To show that $\phi$ satisfies (G1), some minor case work is required. 
Let us define a vector $\vec x \in \Z^n_{\ge 0}$ and two agents $i, j \in N$. 
Let $\vec y$ be the vector resulting from starting at $\vec x$ and adding one unit to $x_i$ and similarly, let $\vec z$ be the allocation resulting from starting at $\vec x$ and adding one unit to $x_j$. 
We need to show that if $\phi(\vec x, i) < \phi(\vec x, j)$ then $\vec z \succ_{\Psi} \vec y$ and if $\phi(\vec x, i) = \phi(\vec x, j)$, then $\vec y =_{\Psi} \vec z$.
\begin{enumerate}[label={\bfseries Case \arabic*:},itemindent=*,leftmargin=0cm]
\item $x_i = x_j = 0$. In this case, it is easy to see that both $\phi(\vec x, i) = \phi(\vec x, j)$ and $\vec y =_{\Psi} \vec z$.
\item $x_i > x_j = 0$. By construction $\phi(\vec x, j) > \phi(\vec x, i)$. We also have $\vec z \succ_{\Psi} \vec y$ since $\vec z$ has fewer indices with the value $0$. Note that this argument also covers the case where $x_j > x_i = 0$.

\item $x_i > 0$  and $x_j > 0$. In this case, note that 
\begin{align*}
   \phi(\vec x, i) = \frac{(x_i + 1)^{w_i}}{{x_i}^{w_i}} = \frac{{y_i}^{w_i}}{{x_i}^{w_i}} = \frac{\prod_{h \in P_{\vec y}} {y_i}^{w_h}}{\prod_{h \in P_{\vec x}} {x_h}^{w_h}}
\end{align*}
where $P_{\vec x}$ denotes the number of non-zero valued indices in $\vec x$. Similarly, $\phi(\vec x, j) = \frac{\prod_{h \in P_{\vec z}} {z_h}^{w_h}}{\prod_{h \in P_{\vec x}} {x_h}^{w_h}}$. 
\end{enumerate}
Therefore, since $|P_{\vec y}| = |P_{\vec z}| > 0$, we have $\phi(\vec x, i) > \phi(\vec x, j)$ if and only if $\vec y \succ_{\Psi} \vec z$. Similarly, we have $\phi(\vec x, i) = \phi(\vec x, j)$ if and only if $\vec y =_{\Psi} \vec z$.
\end{proof}
When all weights are uniform, leximin and max Nash welfare allocations have the same sorted utility vector \citep{Babaioff2021Dichotomous}. 
This implies that they are equivalent notions of fairness. 
However, when agents have different weights, weighted leximin and weighted Nash welfare are not equivalent and can have different sorted utility vectors. Consider the following example:
\begin{example}
We have two agents $N = \{1, 2\}$ and six goods; $w_1 = 2$ and $w_2 = 8$. Both agents have additive valuations, and value all items at $1$.
Any weighted leximin allocation allocates two goods to agent $1$ and four goods to agent $2$, with 
a weighted sorted utility vector of $(\frac12, 1)$. 
However, any \MWNW allocation allocates one good to agent $1$ and five goods to agent $2$. 
This allocation has a worse sorted weighted utility vector $(\frac12, \frac58)$, but a higher weighted Nash welfare. 
\end{example}

\subsection{Max Weighted $p$-Mean Welfare Allocations}\label{sec:p-mean}
The weighted $p$-mean welfare of an allocation is given by $M_p(X) =  \big ( \sum_{i \in N} w_i \times v_i(X_i)^p \big)^{1/p}$ where $p \le 1$. $p$-mean welfare functions have been extensively studied in economics \citep{moulin2004fair} and machine learning \citep{cousins2021axiomatic, cousins2021bounds, heidari2018fairness}. 
To adapt it to fair allocation, we make one minor modification to ensure it is well defined. 
We define a max weighted $p$-mean welfare allocation $X$ as one that first maximizes the number of agents who receive a non-zero utility $P_X$ and subject to that, maximizes $\left( \sum_{i \in P_X} w_i \times v_i(X_i)^p \right)^{1/p}$.

The weighted $p$-mean welfare function as $p$ approaches $0$ corresponds to the weighted Nash welfare of an allocation, and as $p$ approaches $-\infty$, corresponds to the leximin allocation. 
For all the other $p$ values, we can compute max weighted $p$-mean welfare allocations using the following gain function
\begin{align}
    \phi(X, i) = 
    \begin{cases}
        w_i [(v_i(X_i) + 1)^p - v_i(X_i)^p] & p \in (0, 1] \text{ and } v_i(X_i) > 0 \\
        w_i [v_i(X_i)^p - (v_i(X_i) + 1)^p] & p < 0 \text{ and } v_i(X_i) > 0 \\
        M w_i & v_i(X_i) = 0
    \end{cases} \label{eq:phi-weighted-p-mean}
\end{align}
where $M$ is a number greater than any $w_i |(v_i(X_i) + d)^p - v_i(X_i)^p|$.
\begin{restatable}{theorem}{thmweightedpmean}\label{thm:weighted-p-mean}
When agents have MRF valuations and each agent $i$ has a weight $w_i$, General Yankee Swap with $\phi$ given by \eqref{eq:phi-weighted-p-mean} computes a max weighted $p$-mean welfare allocation for any $p \le 1$.
\end{restatable}

\subsection{On the Complexity of Computing $\phi$}
For all the justice criteria discussed above, $b = O(1)$ ($b$ is the size of the vector output by $\phi$), and $\phi$ can be trivially computed in $O(T_v)$ time where $T_v$ is the complexity of computing the value of a bundle. Interestingly, we can trivially speed up the computation of $\phi$ even further to $O(1)$ time. The only queries we make to $\phi$ under General Yankee Swap is with the allocation $X$ maintained by the algorithm. Since this allocation is always non-redundant (Lemma \ref{lem:general-yankee-nonredundant}), we have $v_i(X_i) = |X_i|$ for any agent $i$. Therefore, we can store the sizes of the allocated bundles in $X$ at no extra cost to the time complexity and compute $\phi$ in $O(1)$ time.
\section{Limitations}
The previous section describes several fairness objectives for which Yankee Swap works. This raises the natural question: {\em is there any reasonable fairness notion where Yankee Swap does not work?}

The main limitation of Yankee Swap is that it cannot be used to achieve envy based fairness properties. An allocation $X$ is said to be envy free if $v_i(X_i) > v_i(X_j)$ for all $i, j \in N$. Indeed, this is not always possible to achieve. This impossibility has resulted in several relaxations like envy free up to one good (EF1) \citep{Lipton2004EF1, Budish2011EF1} and envy free up to any good (EFX) \citep{Caragiannis2016MNW, Plaut2017EFX}. However, at its core, these envy relaxations are still fairness objectives that violate the Pareto dominance property (C1): by increasing the utility of an agent currently being envied by other agents, we decrease the fairness of the allocation while Pareto dominating the allocation. One work around for this is using Yankee Swap to compute leximin allocations and hope that leximin allocations have good envy guarantees. This works when all the agents have equal weights --- prioritized Lorenz dominating allocations are guaranteed to be EFX. However, when agents have different weights, Yankee Swap fails to compute weighted envy free up to one good (WEF1) allocations \citep{chakraborty2021weighted}. For clarity, an allocation $X$ is WEF1 when for all $i, j \in N$ $\frac{v_i(X_i)}{w_i} > \frac{v_i(X_j - g)}{w_j}$ for some $g \in X_j$. Yankee Swap fails mainly due to the fact that when agents have MRF valuations, it may be the case that no \MAXUSW allocation is WEF1. Therefore, irrespective of the choice of $\phi$, Yankee Swap cannot always compute a WEF1 allocation (Proposition \ref{prop:yankee-swap-usw}) This is illustrated in the following example.

\begin{example}
We have two agents $\{1, 2\}$ and four goods $\{g_1, \dots, g_4\}$. We have $w_1 = 10$ and $w_2 = 1$. 
The valuation function of every agent $i$ is the MRF $v_i(S) = \min\{|S|, 2\}$.
Any \MAXUSW allocation will assign two goods to both agents. 
However, no such allocation is WEF1: agent $1$ (weighted) envies agent $2$. 
This is because agent $1$'s weighted utility is $\frac{2}{10}$, but agent $2$'s weighted utility is $2$; even after dropping any good, agent $2$'s weighted utility (as seen by both agents) will be $1$. 

There does exist a WEF1 allocation in this example but achieving WEF1 comes at an unreasonable cost of welfare. To achieve WEF1, we must allocate $3$ goods to agent $1$. However, the third good offers no value to agent $1$. In effect, we are obligated to ``burn'' an item in order to satisfy WEF1.
\end{example}

This example seems to suggest that WEF1 is not a suitable envy relaxation for MRF valuations. Indeed, ours is not the first work with this complaint. \citet{chakraborty2022generalized} and \citet{montanari2022weightedenvy} offer alternatives to the WEF1 notion that are more suitable for submodular valuations. Interestingly, as \citet{montanari2022weightedenvy} show, General Yankee Swap can be used to compute a family of these measures called Transferable Weighted Envy Freeness (TWEF) when agents have matroid rank valuations. They show this by first introducing a family of justice criteria called weighted harmonic welfare (parameterized by a variable $x$) which satisfy (C1) and (C2) and then showing that max weighted harmonic welfare allocations are TWEF. 

This suggests that General Yankee Swap can be used to compute allocations which satisfy several envy notions, but finding the appropriate gain function to do so is a more complex task.



\section{Conclusions and Future Work}
In this work, we study fair division of goods when agents have MRF valuations. Our main contribution is a flexible framework that optimizes several fairness objectives. The General Yankee Swap framework is fast, strategyproof and always maximizes utilitarian social welfare.

The General Yankee Swap framework is a strong theoretical tool. We believe it has several applications outside the ones described in Section \ref{sec:applications}. This is a very promising direction for future work.
One specific example is the computation of weighted maximin share allocations. 
Theorem \ref{thm:fair-share} shows that to compute an allocation which gives each agent the maximum fraction of their fair share possible, it suffices to simply input these fair shares into Yankee Swap. Therefore, the problem of computing fair share allocations have effectively been reduced to the problem of computing fair shares when agents have MRF valuations. It would be very interesting to see the different kinds of weighted maximin shares that can be computed and used by Yankee Swap.

Another interesting direction is to show the tightness of this framework. It is unclear whether the conditions (C1) and (C2) are absolutely necessary for Yankee Swap to work. If they are not, it would be interesting to see how they can be relaxed to allow for the optimization of a larger class of fairness properties. 



%
\section*{Acknowledgements}
The authors would like to thank Cyrus Cousins and anonymous reviewers at EC 2023 for useful feedback and suggestions.

\bibliographystyle{ACM-Reference-Format}
\bibliography{abb,literature}


\begin{thebibliography}{33}


\ifx \showCODEN    \undefined \def \showCODEN     #1{\unskip}     \fi
\ifx \showDOI      \undefined \def \showDOI       #1{#1}\fi
\ifx \showISBNx    \undefined \def \showISBNx     #1{\unskip}     \fi
\ifx \showISBNxiii \undefined \def \showISBNxiii  #1{\unskip}     \fi
\ifx \showISSN     \undefined \def \showISSN      #1{\unskip}     \fi
\ifx \showLCCN     \undefined \def \showLCCN      #1{\unskip}     \fi
\ifx \shownote     \undefined \def \shownote      #1{#1}          \fi
\ifx \showarticletitle \undefined \def \showarticletitle #1{#1}   \fi
\ifx \showURL      \undefined \def \showURL       {\relax}        \fi
\providecommand\bibfield[2]{#2}
\providecommand\bibinfo[2]{#2}
\providecommand\natexlab[1]{#1}
\providecommand\showeprint[2][]{arXiv:#2}

\bibitem[\protect\citeauthoryear{Aziz, Chan, and Li}{Aziz
  et~al\mbox{.}}{2019}]%
        {aziz2019wmms}
\bibfield{author}{\bibinfo{person}{Haris Aziz}, \bibinfo{person}{Hau Chan},
  {and} \bibinfo{person}{Bo Li}.} \bibinfo{year}{2019}\natexlab{}.
\newblock \showarticletitle{Weighted Maxmin Fair Share Allocation of
  Indivisible Chores}. In \bibinfo{booktitle}{\emph{Proceedings of the 28th
  International Joint Conference on Artificial Intelligence (IJCAI)}}.
  \bibinfo{pages}{46--52}.
\newblock


\bibitem[\protect\citeauthoryear{Aziz, Moulin, and Sandomirskiy}{Aziz
  et~al\mbox{.}}{2020}]%
        {azis2020wprop}
\bibfield{author}{\bibinfo{person}{Haris Aziz}, \bibinfo{person}{Herv{\'{e}}
  Moulin}, {and} \bibinfo{person}{Fedor Sandomirskiy}.}
  \bibinfo{year}{2020}\natexlab{}.
\newblock \showarticletitle{A polynomial-time algorithm for computing a Pareto
  optimal and almost proportional allocation}.
\newblock \bibinfo{journal}{\emph{Operations Research Letters}}
  \bibinfo{volume}{48}, \bibinfo{number}{5} (\bibinfo{year}{2020}),
  \bibinfo{pages}{573--578}.
\newblock


\bibitem[\protect\citeauthoryear{Babaioff, Ezra, and Feige}{Babaioff
  et~al\mbox{.}}{2021a}]%
        {Babaioff2021Dichotomous}
\bibfield{author}{\bibinfo{person}{Moshe Babaioff}, \bibinfo{person}{Tomer
  Ezra}, {and} \bibinfo{person}{Uriel Feige}.}
  \bibinfo{year}{2021}\natexlab{a}.
\newblock \showarticletitle{Fair and Truthful Mechanisms for Dichotomous
  Valuations}. In \bibinfo{booktitle}{\emph{Proceedings of the 35th AAAI
  Conference on Artificial Intelligence (AAAI)}}. \bibinfo{pages}{5119--5126}.
\newblock


\bibitem[\protect\citeauthoryear{Babaioff, Ezra, and Feige}{Babaioff
  et~al\mbox{.}}{2021b}]%
        {babaioff2021wmms}
\bibfield{author}{\bibinfo{person}{Moshe Babaioff}, \bibinfo{person}{Tomer
  Ezra}, {and} \bibinfo{person}{Uriel Feige}.}
  \bibinfo{year}{2021}\natexlab{b}.
\newblock \showarticletitle{Fair-Share Allocations for Agents with Arbitrary
  Entitlements}. In \bibinfo{booktitle}{\emph{Proceedings of the 22nd ACM
  Conference on Economics and Computation (EC)}}. \bibinfo{pages}{127}.
\newblock


\bibitem[\protect\citeauthoryear{Babaioff and Feige}{Babaioff and
  Feige}{2022}]%
        {babaioff2022fairshare}
\bibfield{author}{\bibinfo{person}{Moshe Babaioff} {and} \bibinfo{person}{Uriel
  Feige}.} \bibinfo{year}{2022}\natexlab{}.
\newblock \showarticletitle{Fair Shares: Feasibility, Domination and
  Incentives}. In \bibinfo{booktitle}{\emph{Proceedings of the 23rd ACM
  Conference on Economics and Computation (EC)}}. \bibinfo{pages}{435}.
\newblock


\bibitem[\protect\citeauthoryear{Barman and Verma}{Barman and Verma}{2021}]%
        {Barman2021MRFMaxmin}
\bibfield{author}{\bibinfo{person}{Siddharth Barman} {and}
  \bibinfo{person}{Paritosh Verma}.} \bibinfo{year}{2021}\natexlab{}.
\newblock \showarticletitle{Existence and Computation of Maximin Fair
  Allocations Under Matroid-Rank Valuations}. In
  \bibinfo{booktitle}{\emph{Proceedings of the 20th International Conference on
  Autonomous Agents and Multi-Agent Systems (AAMAS)}}.
  \bibinfo{pages}{169--177}.
\newblock


\bibitem[\protect\citeauthoryear{Barman and Verma}{Barman and Verma}{2022}]%
        {barman2022groupstrategyproof}
\bibfield{author}{\bibinfo{person}{Siddharth Barman} {and}
  \bibinfo{person}{Paritosh Verma}.} \bibinfo{year}{2022}\natexlab{}.
\newblock \showarticletitle{Truthful and Fair Mechanisms for Matroid-Rank
  Valuations}. In \bibinfo{booktitle}{\emph{Proceedings of the 36th AAAI
  Conference on Artificial Intelligence (AAAI)}}. \bibinfo{pages}{4801--4808}.
\newblock


\bibitem[\protect\citeauthoryear{Benabbou, Chakraborty, Elkind, and
  Zick}{Benabbou et~al\mbox{.}}{2019}]%
        {benabbou2019group}
\bibfield{author}{\bibinfo{person}{Nawal Benabbou}, \bibinfo{person}{Mithun
  Chakraborty}, \bibinfo{person}{Edith Elkind}, {and} \bibinfo{person}{Yair
  Zick}.} \bibinfo{year}{2019}\natexlab{}.
\newblock \showarticletitle{Fairness Towards Groups of Agents in the Allocation
  of Indivisible Items}. In \bibinfo{booktitle}{\emph{Proceedings of the 28th
  International Joint Conference on Artificial Intelligence (IJCAI)}}.
  \bibinfo{pages}{95--101}.
\newblock


\bibitem[\protect\citeauthoryear{Benabbou, Chakraborty, Igarashi, and
  Zick}{Benabbou et~al\mbox{.}}{2021}]%
        {benabbou2021MRF}
\bibfield{author}{\bibinfo{person}{Nawal Benabbou}, \bibinfo{person}{Mithun
  Chakraborty}, \bibinfo{person}{Ayumi Igarashi}, {and} \bibinfo{person}{Yair
  Zick}.} \bibinfo{year}{2021}\natexlab{}.
\newblock \showarticletitle{Finding Fair and Efficient Allocations for Matroid
  Rank Valuations}.
\newblock \bibinfo{journal}{\emph{{ACM} Transactions on Economics and
  Computation}} \bibinfo{volume}{9}, \bibinfo{number}{4}, Article
  \bibinfo{articleno}{21} (\bibinfo{year}{2021}).
\newblock


\bibitem[\protect\citeauthoryear{Bouveret, Chevaleyre, and Maudet}{Bouveret
  et~al\mbox{.}}{2016}]%
        {bouveret2016handbook}
\bibfield{author}{\bibinfo{person}{Sylvain Bouveret}, \bibinfo{person}{Yann
  Chevaleyre}, {and} \bibinfo{person}{Nicolas Maudet}.}
  \bibinfo{year}{2016}\natexlab{}.
\newblock \showarticletitle{Fair Allocation of Indivisible Goods}.
\newblock In \bibinfo{booktitle}{\emph{Handbook of Computational Social
  Choice}}, \bibfield{editor}{\bibinfo{person}{Felix Brandt},
  \bibinfo{person}{Vincent Conitzer}, \bibinfo{person}{Ulle Endriss},
  \bibinfo{person}{J{\'e}r{\^o}me Lang}, {and} \bibinfo{person}{Ariel~D.
  Procaccia}} (Eds.). \bibinfo{publisher}{Cambridge University Press},
  Chapter~12.
\newblock


\bibitem[\protect\citeauthoryear{Budish}{Budish}{2011}]%
        {Budish2011EF1}
\bibfield{author}{\bibinfo{person}{Eric Budish}.}
  \bibinfo{year}{2011}\natexlab{}.
\newblock \showarticletitle{The Combinatorial Assignment Problem: Approximate
  Competitive Equilibrium from Equal Incomes.}
\newblock \bibinfo{journal}{\emph{Journal of Political Economy}}
  \bibinfo{volume}{119}, \bibinfo{number}{6} (\bibinfo{year}{2011}),
  \bibinfo{pages}{1061 -- 1103}.
\newblock


\bibitem[\protect\citeauthoryear{Caragiannis, Kurokawa, Moulin, Procaccia,
  Shah, and Wang}{Caragiannis et~al\mbox{.}}{2016}]%
        {Caragiannis2016MNW}
\bibfield{author}{\bibinfo{person}{Ioannis Caragiannis}, \bibinfo{person}{David
  Kurokawa}, \bibinfo{person}{Herv\'{e} Moulin}, \bibinfo{person}{Ariel~D.
  Procaccia}, \bibinfo{person}{Nisarg Shah}, {and} \bibinfo{person}{Junxing
  Wang}.} \bibinfo{year}{2016}\natexlab{}.
\newblock \showarticletitle{The Unreasonable Fairness of Maximum Nash Welfare}.
  In \bibinfo{booktitle}{\emph{Proceedings of the 17th ACM Conference on
  Economics and Computation (EC)}}. \bibinfo{pages}{305–322}.
\newblock


\bibitem[\protect\citeauthoryear{Chakrabarty, Tat~Lee, Sidford, Singla, and
  Chiu-wai Wong}{Chakrabarty et~al\mbox{.}}{2019}]%
        {Chakrabarty2019MatroidIntersection}
\bibfield{author}{\bibinfo{person}{Deeparnab Chakrabarty}, \bibinfo{person}{Yin
  Tat~Lee}, \bibinfo{person}{Aaron Sidford}, \bibinfo{person}{Sahil Singla},
  {and} \bibinfo{person}{Sam Chiu-wai Wong}.} \bibinfo{year}{2019}\natexlab{}.
\newblock \showarticletitle{Faster Matroid Intersection}. In
  \bibinfo{booktitle}{\emph{Proceedings of the 60th Symposium on Foundations of
  Computer Science (FOCS)}}. \bibinfo{pages}{1146--1168}.
\newblock


\bibitem[\protect\citeauthoryear{Chakraborty, Igarashi, Suksompong, and
  Zick}{Chakraborty et~al\mbox{.}}{2021a}]%
        {chakraborty2021weighted}
\bibfield{author}{\bibinfo{person}{Mithun Chakraborty}, \bibinfo{person}{Ayumi
  Igarashi}, \bibinfo{person}{Warut Suksompong}, {and} \bibinfo{person}{Yair
  Zick}.} \bibinfo{year}{2021}\natexlab{a}.
\newblock \showarticletitle{Weighted Envy-Freeness in Indivisible Item
  Allocation}.
\newblock \bibinfo{journal}{\emph{ACM Transactions on Economics and
  Computation}}  \bibinfo{volume}{9} (\bibinfo{year}{2021}).
\newblock
\showISSN{2167-8375}


\bibitem[\protect\citeauthoryear{Chakraborty, Schmidt-Kraepelin, and
  Suksompong}{Chakraborty et~al\mbox{.}}{2021b}]%
        {chakraborty2021pickingsequences}
\bibfield{author}{\bibinfo{person}{Mithun Chakraborty}, \bibinfo{person}{Ulrike
  Schmidt-Kraepelin}, {and} \bibinfo{person}{Warut Suksompong}.}
  \bibinfo{year}{2021}\natexlab{b}.
\newblock \showarticletitle{Picking sequences and monotonicity in weighted fair
  division}.
\newblock \bibinfo{journal}{\emph{Artificial Intelligence}}
  \bibinfo{volume}{301} (\bibinfo{year}{2021}).
\newblock


\bibitem[\protect\citeauthoryear{Chakraborty, Segal-Halevi, and
  Suksompong}{Chakraborty et~al\mbox{.}}{2022}]%
        {chakraborty2022generalized}
\bibfield{author}{\bibinfo{person}{Mithun Chakraborty}, \bibinfo{person}{Erel
  Segal-Halevi}, {and} \bibinfo{person}{Warut Suksompong}.}
  \bibinfo{year}{2022}\natexlab{}.
\newblock \showarticletitle{Weighted Fairness Notions for Indivisible Items
  Revisited}.
\newblock \bibinfo{journal}{\emph{Proceedings of the 36th AAAI Conference on
  Artificial Intelligence (AAAI)}} (\bibinfo{year}{2022}),
  \bibinfo{pages}{4949--4956}.
\newblock


\bibitem[\protect\citeauthoryear{Chiarelli, Krnc, Milani{\v c}, Pferschy,
  Piva{\v c}, and Schauer}{Chiarelli et~al\mbox{.}}{2022}]%
        {chiarelli2022fairnessconflicts}
\bibfield{author}{\bibinfo{person}{Nina Chiarelli}, \bibinfo{person}{Matja{\v
  z} Krnc}, \bibinfo{person}{Martin Milani{\v c}}, \bibinfo{person}{Ulrich
  Pferschy}, \bibinfo{person}{Nevena Piva{\v c}}, {and}
  \bibinfo{person}{Joachim Schauer}.} \bibinfo{year}{2022}\natexlab{}.
\newblock \showarticletitle{Fair Allocation of Indivisible Items with Conflict
  Graphs}.
\newblock \bibinfo{journal}{\emph{Algorithmica}} (\bibinfo{year}{2022}).
\newblock


\bibitem[\protect\citeauthoryear{Cousins}{Cousins}{2021a}]%
        {cousins2021axiomatic}
\bibfield{author}{\bibinfo{person}{Cyrus Cousins}.}
  \bibinfo{year}{2021}\natexlab{a}.
\newblock \showarticletitle{An Axiomatic Theory of Provably-Fair
  Welfare-Centric Machine Learning}. In \bibinfo{booktitle}{\emph{Proceedings
  of the 35th Annual Conference on Neural Information Processing Systems
  (NeurIPS)}}. \bibinfo{pages}{16610--16621}.
\newblock


\bibitem[\protect\citeauthoryear{Cousins}{Cousins}{2021b}]%
        {cousins2021bounds}
\bibfield{author}{\bibinfo{person}{Cyrus Cousins}.}
  \bibinfo{year}{2021}\natexlab{b}.
\newblock \emph{\bibinfo{title}{Bounds and Applications of Concentration of
  Measure in Fair Machine Learning and Data Science}}.
\newblock \bibinfo{thesistype}{Ph.D. Dissertation}. \bibinfo{school}{Brown
  University}.
\newblock


\bibitem[\protect\citeauthoryear{Farhadi, Ghodsi, Hajiaghayi, Lahaie, Pennock,
  Seddighin, Seddighin, and Yami}{Farhadi et~al\mbox{.}}{2019}]%
        {farhadi2019wmms}
\bibfield{author}{\bibinfo{person}{Alireza Farhadi}, \bibinfo{person}{Mohammad
  Ghodsi}, \bibinfo{person}{MohammadTaghi Hajiaghayi},
  \bibinfo{person}{S\'{e}bastien Lahaie}, \bibinfo{person}{David Pennock},
  \bibinfo{person}{Masoud Seddighin}, \bibinfo{person}{Saeed Seddighin}, {and}
  \bibinfo{person}{Hadi Yami}.} \bibinfo{year}{2019}\natexlab{}.
\newblock \showarticletitle{Fair Allocation of Indivisible Goods to Asymmetric
  Agents}.
\newblock \bibinfo{journal}{\emph{Journal of Artificial Intelligence Research}}
  \bibinfo{volume}{64}, \bibinfo{number}{1} (\bibinfo{year}{2019}),
  \bibinfo{pages}{1–20}.
\newblock


\bibitem[\protect\citeauthoryear{Garg, Husić, Murhekar, and Végh}{Garg
  et~al\mbox{.}}{2021}]%
        {garg2021mwnw}
\bibfield{author}{\bibinfo{person}{Jugal Garg}, \bibinfo{person}{Edin Husić},
  \bibinfo{person}{Aniket Murhekar}, {and} \bibinfo{person}{László Végh}.}
  \bibinfo{year}{2021}\natexlab{}.
\newblock \bibinfo{title}{Tractable Fragments of the Maximum Nash Welfare
  Problem}.
\newblock
\newblock
\urldef\tempurl%
\url{https://doi.org/10.48550/ARXIV.2112.10199}
\showDOI{\tempurl}


\bibitem[\protect\citeauthoryear{Ghodsi, HajiAghayi, Seddighin, Seddighin, and
  Yami}{Ghodsi et~al\mbox{.}}{2018}]%
        {ghodsi2018fair}
\bibfield{author}{\bibinfo{person}{Mohammad Ghodsi},
  \bibinfo{person}{MohammadTaghi HajiAghayi}, \bibinfo{person}{Masoud
  Seddighin}, \bibinfo{person}{Saeed Seddighin}, {and} \bibinfo{person}{Hadi
  Yami}.} \bibinfo{year}{2018}\natexlab{}.
\newblock \showarticletitle{Fair allocation of indivisible goods: Improvements
  and generalizations}. In \bibinfo{booktitle}{\emph{Proceedings of the 19th
  ACM Conference on Economics and Computation (EC)}}.
  \bibinfo{pages}{539--556}.
\newblock


\bibitem[\protect\citeauthoryear{Halpern, Procaccia, Psomas, and Shah}{Halpern
  et~al\mbox{.}}{2020}]%
        {halpern2020binaryadditive}
\bibfield{author}{\bibinfo{person}{Daniel Halpern}, \bibinfo{person}{Ariel~D.
  Procaccia}, \bibinfo{person}{Alexandros Psomas}, {and}
  \bibinfo{person}{Nisarg Shah}.} \bibinfo{year}{2020}\natexlab{}.
\newblock \showarticletitle{Fair Division with Binary Valuations: One Rule to
  Rule Them All}. In \bibinfo{booktitle}{\emph{Proceedings of the 16th
  Conference on Web and Internet Economics (WINE)}}.
  \bibinfo{pages}{370–383}.
\newblock


\bibitem[\protect\citeauthoryear{Heidari, Ferrari, Gummadi, and Krause}{Heidari
  et~al\mbox{.}}{2018}]%
        {heidari2018fairness}
\bibfield{author}{\bibinfo{person}{Hoda Heidari}, \bibinfo{person}{Claudio
  Ferrari}, \bibinfo{person}{Krishna Gummadi}, {and} \bibinfo{person}{Andreas
  Krause}.} \bibinfo{year}{2018}\natexlab{}.
\newblock \showarticletitle{Fairness behind a veil of ignorance: A welfare
  analysis for automated decision making}. In
  \bibinfo{booktitle}{\emph{Advances in Neural Information Processing
  Systems}}. \bibinfo{pages}{1265--1276}.
\newblock


\bibitem[\protect\citeauthoryear{Li, Li, and Wu}{Li et~al\mbox{.}}{2022}]%
        {li2022weightedproportional}
\bibfield{author}{\bibinfo{person}{Bo Li}, \bibinfo{person}{Yingkai Li}, {and}
  \bibinfo{person}{Xiaowei Wu}.} \bibinfo{year}{2022}\natexlab{}.
\newblock \showarticletitle{Almost (Weighted) Proportional Allocations for
  Indivisible Chores}. In \bibinfo{booktitle}{\emph{Proceedings of the ACM Web
  Conference 2022}}. \bibinfo{pages}{122–131}.
\newblock


\bibitem[\protect\citeauthoryear{Lipton, Markakis, Mossel, and Saberi}{Lipton
  et~al\mbox{.}}{2004}]%
        {Lipton2004EF1}
\bibfield{author}{\bibinfo{person}{R.~J. Lipton}, \bibinfo{person}{E.
  Markakis}, \bibinfo{person}{E. Mossel}, {and} \bibinfo{person}{A. Saberi}.}
  \bibinfo{year}{2004}\natexlab{}.
\newblock \showarticletitle{On Approximately Fair Allocations of Indivisible
  Goods}. In \bibinfo{booktitle}{\emph{Proceedings of the 5th ACM Conference on
  Economics and Computation (EC)}}. \bibinfo{pages}{125–131}.
\newblock


\bibitem[\protect\citeauthoryear{Montanari, Schmidt-Kraepelin, Suksompong, and
  Teh}{Montanari et~al\mbox{.}}{2022}]%
        {montanari2022weightedenvy}
\bibfield{author}{\bibinfo{person}{Luisa Montanari}, \bibinfo{person}{Ulrike
  Schmidt-Kraepelin}, \bibinfo{person}{Warut Suksompong}, {and}
  \bibinfo{person}{Nicholas Teh}.} \bibinfo{year}{2022}\natexlab{}.
\newblock \bibinfo{title}{Weighted Envy-Freeness for Submodular Valuations}.
\newblock
\newblock
\urldef\tempurl%
\url{https://doi.org/10.48550/ARXIV.2209.06437}
\showDOI{\tempurl}


\bibitem[\protect\citeauthoryear{Moulin}{Moulin}{2004}]%
        {moulin2004fair}
\bibfield{author}{\bibinfo{person}{Herv{\'e} Moulin}.}
  \bibinfo{year}{2004}\natexlab{}.
\newblock \bibinfo{booktitle}{\emph{Fair division and collective welfare}}.
\newblock \bibinfo{publisher}{MIT Press}.
\newblock


\bibitem[\protect\citeauthoryear{Plaut and Roughgarden}{Plaut and
  Roughgarden}{2017}]%
        {Plaut2017EFX}
\bibfield{author}{\bibinfo{person}{Benjamin Plaut} {and} \bibinfo{person}{Tim
  Roughgarden}.} \bibinfo{year}{2017}\natexlab{}.
\newblock \showarticletitle{Almost Envy-Freeness with General Valuations}.
\newblock \bibinfo{journal}{\emph{ArXiv}}  \bibinfo{volume}{abs/1707.04769}
  (\bibinfo{year}{2017}).
\newblock


\bibitem[\protect\citeauthoryear{Procaccia and Wang}{Procaccia and
  Wang}{2014}]%
        {procaccia2014fairenough}
\bibfield{author}{\bibinfo{person}{Ariel~D. Procaccia} {and}
  \bibinfo{person}{Junxing Wang}.} \bibinfo{year}{2014}\natexlab{}.
\newblock \showarticletitle{Fair Enough: Guaranteeing Approximate Maximin
  Shares}. In \bibinfo{booktitle}{\emph{Proceedings of the 15th ACM Conference
  on Economics and Computation (EC)}}. \bibinfo{pages}{675--692}.
\newblock


\bibitem[\protect\citeauthoryear{Schrijver}{Schrijver}{2003}]%
        {schrijver-book}
\bibfield{author}{\bibinfo{person}{A. Schrijver}.}
  \bibinfo{year}{2003}\natexlab{}.
\newblock \bibinfo{booktitle}{\emph{Combinatorial Optimization - Polyhedra and
  Efficiency}}.
\newblock \bibinfo{publisher}{Springer}.
\newblock


\bibitem[\protect\citeauthoryear{Suksompong and Teh}{Suksompong and
  Teh}{2022}]%
        {Suksumpong2022weightednash}
\bibfield{author}{\bibinfo{person}{Warut Suksompong} {and}
  \bibinfo{person}{Nicholas Teh}.} \bibinfo{year}{2022}\natexlab{}.
\newblock \showarticletitle{On maximum weighted Nash welfare for binary
  valuations}.
\newblock \bibinfo{journal}{\emph{Mathematical Social Sciences}}
  \bibinfo{volume}{117} (\bibinfo{year}{2022}), \bibinfo{pages}{101--108}.
\newblock
\showISSN{0165-4896}


\bibitem[\protect\citeauthoryear{Viswanathan and Zick}{Viswanathan and
  Zick}{2023}]%
        {viswanathan2022yankee}
\bibfield{author}{\bibinfo{person}{Vignesh Viswanathan} {and}
  \bibinfo{person}{Yair Zick}.} \bibinfo{year}{2023}\natexlab{}.
\newblock \showarticletitle{Yankee Swap: a Fast and Simple Fair Allocation
  Mechanism for Matroid Rank Valuations}. In
  \bibinfo{booktitle}{\emph{Proceedings of the 22nd International Conference on
  Autonomous Agents and Multi-Agent Systems (AAMAS)}}.
\newblock


\end{thebibliography}

\newpage
\appendix

\section{Missing Proofs from Section \ref{sec:strategyproofness}}
\lemmonotonicity*

\begin{proof}
Define the new valuation functions by $\{v'_j\}_{j \in N}$ where $v'_j = v_j$ for all $j \in N - i$. To prevent any ambiguity, whenever we discuss the $\Psi$ value of $X$, we implicitly discuss it with respect to the valuations $v$. Similarly,  whenever we discuss the $\Psi$ value of $Y$ we implicitly discuss it with respect to the valuations $v'$.

Assume for contradiction that $|Y_i| < |X_i|$. Let $T$ be a subset of $Y_i$ such that $|T| = v_i(T) = v_i(Y_i)$. Define an allocation $Z$ as $Z_h = Y_h$ for all $h \in N - i$ and $Z_i = T$; allocate the remaining goods in $Z$ to $Z_0$. Note that both $X$ and $Z$ are non-redundant under both valuation profiles. This implies for all $j \in N$, $v'_j(X_j) = v_j(X_j) = |X_j|$ and $v'_j(Z_j) = v_j(Z_j) = |Z_j|$.

Let us compare $X$ and $Y$. Let $p \in N$ be the agent with highest $\phi(Y, p)$ such that $|Y_p| < |X_p|$; break ties by choosing the least $p$. Such an element is guaranteed to exist since $|Y_i| < |X_i|$.

If there exists no $q \in N$ such that $|Y_q| > |X_q|$, we must have $|Y_0| > |X_0|$ (since $|X_i| > |Y_i|$). Using Lemma \ref{lem:augmentation-sufficient} with the allocations $Y$ and $X$ with $p$, we get that there is a path from $p$ to $0$ in the exchange graph of $Y$. Transferring goods along the shortest such path results in an allocation with a higher USW than $Y$ contradicting the fact that $Y$ is \MAXUSW.

Let $q \in N$ be the agent with highest $\phi(X, q)$ such that $|Y_q| > |X_q|$. Break ties by choosing the least $q$. Note that since $Z$ and $Y$ only differ in $i$'s bundle and $|Y_i| < |X_i|$, we must have $|Y_q| = |Z_q|$. Therefore, $|Z_q| > |X_q|$ as well.

Consider two cases:
\begin{enumerate}[(i)]
    \item $\phi(X, q) > \phi(Y, p)$,
    \item $\phi(X, q) = \phi(Y, p)$ and $q < p$
\end{enumerate}
If any of the above two conditions occur, then invoking Lemma \ref{lem:augmentation-sufficient} with allocations $X$, $Z$ and the agent $q$, there exists a transfer path from $q$ to some agent $k$ where $|Z_k| < |X_k|$. Transferring along the shortest such path gives us a non redundant (w.r.t. the valuation profile $v$) allocation $X'$ where $|X'_q| = |X_q| + 1$ and $|X'_k| = |X_k| - 1$ (Lemma \ref{lem:path-augmentation}). 
Since $v'_j(X'_j) \ge v_j(X'_j)$ for all $j \in N$, we must have that $X'$ is non-redundant with respect to both valuation profiles $v$ and $v'$.
By our definition of $Z$, if $k \in N - i$, we have $|X_k| > |Z_k| = |Y_k|$ and if $k = i$, we have $|X_i| > |Y_i|$. Therefore, we have $|X_k| > |Y_k|$ if $k \ne 0$.

If $k = 0$, we improve USW (w.r.t. the valuation profile $v$) contradicting the fact that $X$ is $\Psi$ maximizing. 

Let $X''$ be an allocation obtained from starting at $X$ and removing a good from $k$. From (G1), we have that if $\phi(X'', q) \ge \phi(X'', k)$, we have $X' \succeq_{\Psi} X$ (with respect to the valuations $v$) with equality holding if and only if $\phi(X'', k) = \phi(X'', q)$. 

$X, X', X''$ and $Y$ are non-redundant with respect to the new valuation function profile $v'$. We can therefore, compare their $\phi$ values using (G2).

\textbf{For case (i):} If $k \ne 0$, we have  $\phi(X'', k) = \phi(X', k) \le \phi(Y, k) \le \phi(Y, p) < \phi(X, q) = \phi(X'', q)$. This gives us $\phi(X'', k) < \phi(X'', q)$: a contradiction.

\textbf{For case (ii):} If $k \ne 0$, we have  $\phi(X'', k) \le \phi(Y, k) \le \phi(Y, p) = \phi(X, q) = \phi(X'', q)$. If any of these weak inequalities are strict, we can use analysis similar to that of case (i) to show that $X$ does not maximize $\Psi$. Therefore, all the weak inequalities must be equalities and we must have $\phi(X'', k) = \phi(Y, k) = \phi(Y, p) = \phi(X, q) = \phi(X'', q)$.

Since $\phi(X'', k) = \phi(X'', p)$, we have $X' =_{\Psi} X$. 
Moreover, by our choice of $p$ we have $p \le k$ and by assumption, we have $q < p$. Combining the two, this gives us $q < k$. Therefore, $X'$ lexicographically dominates $X$ --- a contradiction to Theorem \ref{thm:weighted-yankee-leximin}.

Let us now move on to the remaining two possible cases

\begin{enumerate}[(i)]\addtocounter{enumi}{2}
    \item $\phi(X, q) < \phi(Y, p)$,
    \item $\phi(X, q) = \phi(Y, p)$ and $q > p$
\end{enumerate}

Recall that $X$ and $Y$ are non-redundant with respect to the new valuation functions $v'$.
If any of the above two conditions occur, then invoking Lemma \ref{lem:augmentation-sufficient} with allocations $Y$, $X$ and the agent $p$, there exists a transfer path from $p$ to some agent $l$ in the exchange graph of $Y$ where $|Y_l| > |X_l|$. Transferring along the shortest such path gives us a non-redundant allocation $Y'$ (according to the valuation $v'$) where $|Y'_p| = |Y_p| + 1$ and $|Y'_l| = |Y_l| - 1$ (Lemma \ref{lem:path-augmentation}). 

Let $Y''$ be an allocation obtained from starting at $Y$ and removing a good from $l$. From (G1), we have that if $\phi(Y'', p) \ge \phi(Y'', l)$, we have $Y' \succeq_{\Psi} Y$ (with respect to the valuations $v'$) with equality holding if and only if $\phi(Y'', l) = \phi(Y'', p)$. 

If $l = 0$, we improve USW (according to the valuations $v'$) contradicting the fact that $Y$ is \MAXUSW. 

\textbf{For case (iii):} If $l \ne 0$, we have $\phi(Y'', l)  \le \phi(X, l) \le \phi(X, q) < \phi(Y, p) = \phi(Y'', p)$. Compressing the inequality, we get that $\phi(Y'', l) < \phi(Y'', p)$: a contradiction.

\textbf{For case (iv):} If $l \ne 0$, we have  $\phi(Y'', l)  \le \phi(X, l) \le \phi(X, q) = \phi(Y, p) = \phi(Y'', p)$. If any of these weak inequalities are strict, we can use analysis similar to that of case (iii) to show that $Y$ does not maximize $\Psi$. Therefore, all the weak inequalities must be equalities and we must have $\phi(Y'', l)  = \phi(X, l) = \phi(X, q) = \phi(Y, p) = \phi(Y'', p)$. This implies that $Y' =_{\Psi} Y$. 

Moreover, by our choice of $q$ we have $q \le l$ and by assumption, we have $p < q$. Combining the two, this gives us $p < l$. Therefore, $Y'$ lexicographically dominates $Y$ --- a contradiction to Theorem \ref{thm:weighted-yankee-leximin}.

Since cases (i)--(iv) cover all possible cases, our proof is complete.
\end{proof}

\section{Missing Proofs from Section \ref{sec:time-complexity}}

\lemfinddesired*
\begin{proof}
Since agents have MRF valuations, if $v_i(S \cup B) > v_i(S)$ there must be some good $g \in B$ such that $\Delta_i(S, g) = 1$. Such a good $g$ must be either in $B_1$ or in $B_2$. Therefore, either $v_i(S \cup B_1) > v_i(S)$ or $v_i(S \cup B_2) > v_i(S)$. So the procedure must be correct. The time complexity comes from the fact that the set $B$ gets halved at every iteration and all other steps in the iteration can be done in $O(m) = O(T_v)$ time.
\end{proof}

\lemgetdistances*
\begin{proof}
Note that the procedure $\FindDesired(i, X_i, B)$ outputs goods in $B$ that the node $s$ has outgoing edges to in $\cal G(X)$; that is, goods in $F_i(X) \cap B$. Similarly, when the allocation $X$ is non-redundant, for some good $a \in X_j$, the procedure $\FindDesired(j, X_j - a, B)$ outputs goods in $B$ that the node $a$ has edges to in $\cal G(X)$. The correctness then follows from the breadth first search approach of the procedure. 

Each in node is added only once to $Q$ and removed only once from $B$. Therefore, $\FindDesired$ is called only $O(m)$ times. Since, all other steps can be done in $O(1)$ time, the total complexity of this algorithm is $O(m T_v \log m)$ (from Lemma \ref{lem:find-desired}).
\end{proof}

\section{Missing Proofs from Section \ref{sec:applications}}

\thmweightedleximin*
\begin{proof}
Formally, for any two allocations $X \succ_{\Psi} Y$ if $\vec e^X$ lexicographically dominates $\vec e^Y$. It is easy to see that $\Psi$ trivially satisfies (C1) and $\phi$ trivially satisfies (G2); so we only show (G1).

Assume $\phi(\vec x, i) > \phi(\vec x, j)$ for some agents $i, j \in N$ and vector $\vec x \in \Z^n_{\ge 0}$. Let $\vec y$ be the allocation that results from starting at $\vec x$ adding one unit to $x_i$ and let $\vec z$ be the allocation that results from starting at $\vec x$ adding one unit to $x_j$.

If $\phi(\vec x, i) > \phi(\vec x, j)$, then one of the following two cases must be true.

\textbf{Case 1:} $\frac{x_i}{w_i} < \frac{x_j}{w_j}$. Since it is always better to add utility to agents with lower weighted utility, we have $\vec y \succ_{\Psi} \vec z$.

\textbf{Case 2:} $\frac{x_i}{w_i}= \frac{x_j}{w_j}$ and $w_i < w_j$. If this is true, we have $\frac{y_j}{w_j} = \frac{z_i}{w_i}$ by assumption. However, $\frac{y_i}{w_i} = \frac{x_i + 1}{w_i} > \frac{x_j + 1}{w_j} = \frac{z_j}{w_j}$. Since the two vectors differ only in the values of the indices $j$ and $i$, we can conclude that $\vec y \succ_{\Psi} \vec z$.

This implies that when $\phi(\vec x, i) > \phi(\vec x, j)$, we have $\vec y \succ_{\Psi} \vec z$ as required. 

When $\phi(\vec x, i) = \phi(\vec x, j)$, we must have $\frac{x_i}{w_j}= \frac{x_j}{w_j}$ and $w_i = w_j$. This gives us $\frac{y_j}{w_j} = \frac{z_i}{w_i}$ and $\frac{y_i}{w_i} = \frac{x_i + 1}{w_i} = \frac{x_j + 1}{w_j} = \frac{z_j}{w_j}$ which implies that $\vec y =_{\Psi} \vec z$. 
\end{proof}

\thmweightedpmean*
\begin{proof}
Formally, $X \succ_{\Psi} Y$  if and only if one of the following conditions hold:
\begin{enumerate}[(a)]
    \item $|P_X| > |P_Y|$ 
    \item $|P_X| = |P_Y|$ and $\left ( \sum_{i \in P_X} w_i \times v_i(X_i)^p \right )^{1/p} > \left ( \sum_{i \in P_Y} w_i \times v_i(Y_i)^p \right)^{1/p}$
\end{enumerate}
It is easy to see that $\Psi$ satisfies (C1) and $\phi$ satisfies (G2).  Similar to the previous results, we only show (G1). For any vector $\vec x \in \Z^n_{\ge 0}$, consider two agents $i$ and $j$. Let $\vec y$ be the vector that results from starting at $\vec x$ and adding $1$ to $x_i$ and $\vec z$ be the vector that results from starting at $\vec x$ and adding $1$ to $x_j$. Let $P_{\vec x}$ be the number of indices in $\vec x$ with a positive value. We have the following three cases.
\begin{enumerate}[label={\bfseries Case \arabic*:},itemindent=*,leftmargin=0cm]
\item $x_i \ge x_j = 0$. 
The proof for this case is the same as those of Cases 1 and 2 in Theorem \ref{thm:weighted-nash}.
\item $p \in (0, 1]$ and $x_i, x_j > 0$. 
We have $P_{\vec y} = P_{\vec z} = P_{\vec x}$. This gives us:
\begin{align*}
    \vec y \succeq_{\Psi} \vec z
    \Leftrightarrow  \sum_{i \in P_{\vec x}} w_iy_i^p  \ge \sum_{i \in P_{\vec x}} w_i z_i^p 
    \Leftrightarrow \sum_{i \in P_{\vec x}} w_i y_i^p - \sum_{i \in P_{\vec x}} w_i x_i^p \ge \sum_{i \in P_{\vec x}} w_i z_i^p - \sum_{i \in P_{\vec x}} w_i x_i^p \\
    \Leftrightarrow w_i [(x_i + 1)^p - x_i^p] \ge w_i [(x_j + 1)^p - x_j^p] \\
    \Leftrightarrow \phi(\vec x, i) \ge \phi(\vec x, j)
\end{align*}

\item $p < 0$ and $x_i, x_j > 0$. We still have $P_{\vec y} = P_{\vec z} = P_{\vec x}$. We have:
\begin{align*}
    \vec y \succeq_{\Psi} \vec z
    \Leftrightarrow  \sum_{i \in P_{\vec x}} w_i y_i^p  \le \sum_{i \in P_{\vec x}} w_i z_i^p 
    \Leftrightarrow \sum_{i \in P_{\vec x}} w_i y_i^p - \sum_{i \in P_{\vec x}} w_i x_i^p \le \sum_{i \in P_{\vec x}} w_i z_i^p - \sum_{i \in P_{\vec x}} w_i x_i^p \\
    \Leftrightarrow w_i [(x_i + 1)^p - x_i^p] \le w_j [(x_j + 1)^p - x_j^p] \\
    \Leftrightarrow \phi(\vec x, i) \ge \phi(\vec x, j)
\end{align*}
In both Cases 2 and 3, we can replace the inequalities with equalities.
\end{enumerate}
\end{proof}

\end{document}